\documentclass[11pt]{article}

\usepackage{euler} 
\usepackage{palatino} 

\usepackage{authblk}
\usepackage[utf8]{inputenc}
\usepackage[T1]{fontenc} 
\usepackage{hyperref}
\usepackage{url} 
\usepackage{booktabs}
\usepackage{amsfonts}
\usepackage{nicefrac}
\usepackage{microtype}
\usepackage{lipsum}
\usepackage{enumitem}
\usepackage[ruled,vlined]{algorithm2e}

\usepackage{xcolor}
\definecolor{ourred}{rgb}{0.90,0.10,0.10}
\usepackage{hyperref}
\hypersetup{
colorlinks=true,
breaklinks=true,
urlcolor= ourred,
linkcolor= ourred,
citecolor= ourred,
}
\usepackage{fullpage}

\SetAlFnt{\small}
\usepackage{graphicx}
\usepackage{subcaption}
\usepackage{amsmath}
\usepackage{mathrsfs}
\usepackage{amssymb}
\usepackage{amsthm}
\newtheorem{prop}{Proposition}

\usepackage{appendix}
\usepackage{chngcntr}
\usepackage{apptools}
\AtAppendix{\counterwithin{prop}{section}}
\AtAppendix{\counterwithin{defn}{section}}

\usepackage{color}
\usepackage{setspace}
\usepackage{lineno,hyperref}

\title{
  \textsc{Characterization of Random Walks\\
  on Space of Unordered Trees\\
  using Efficient Metric Simulation}
}

\author[1]{Farah Ben Naoum}
\author[2]{Christophe Godin}
\author[2,*]{Romain Aza\"is}

\affil[1]{EEDIS Laboratory, Computer Sciences Department, Djillali Liabes University, Sidi Bel-Abbes, Algeria.}
\affil[2]{Laboratoire Reproduction et D\'eveloppement des Plantes, Univ Lyon, ENS de Lyon, UCB Lyon 1, CNRS, INRAE, Inria, F-69342, Lyon, France.}
\affil[*]{Corresponding author: Romain Aza\"is, romain.azais@inria.fr}

\date{}

\DeclareMathOperator*{\argmin}{arg\,min}
\DeclareMathOperator*{\argmax}{arg\,max}
\newcommand{\Root}{\mathcal{R}} 
\newcommand{\Leaves}{\mathcal{L}} 
\newcommand{\Children}{\mathcal{C}} 
\newcommand{\Parent}{\mathcal{P}} 
\newcommand{\Height}{\mathcal{H}} 
\newcommand{\Outdegree}{\mathcal{D}} 
\newcommand{\Strahler}{\mathcal{S}} 
\newcommand{\Dist}{\delta} 
\newcommand{\Neighbors}{\mathcal{N}} 
\newcommand{\Size}{\#} 

\begin{document}

\maketitle

\begin{abstract}
The simple random walk on $\mathbb{Z}^p$ shows two drastically different behaviours depending on the value of $p$: it is recurrent when $p\in\{1,2\}$ while it escapes (with a rate increasing with $p$) as soon as $p\geq3$. This classical example illustrates that the asymptotic properties of a random walk provides some information on the structure of its state space. This paper aims to explore analogous questions on space made up of combinatorial objects with no algebraic structure. We take as a model for this problem the space of unordered unlabeled rooted trees endowed with Zhang edit distance. To this end, it defines the canonical unbiased random walk on the space of trees and provides an efficient algorithm to evaluate its escape rate. Compared to Zhang algorithm, it is incremental and computes the edit distance along the random walk approximately 100 times faster on trees of size $500$ on average. The escape rate of the random walk on trees is precisely estimated using intensive numerical simulations, out of reasonable reach without the incremental algorithm.
\end{abstract}

\medskip

\section{Introduction}

    A random walk is a discrete trajectory, i.e. a sequence of locations in a space, where each location is obtained by a random elementary move from the previous one.
    When the space is endowed with a graph structure, the canonical random walk evolves along the edges of the graph.
    On a discrete metric space, the natural graph structure puts edges between neighbours, i.e. between elements $x$ and $y$ of $\mathcal{X}$ such that $d(x,y)$ is minimal nonnegative, which can be assumed to be $1$ without loss of generality. In addition, the walk is said isotropic when all the elementary moves share the same probability. In other words, the isotropic random walk on $(\mathcal{X},d)$ is a sequence $(X_h)_{h\geq0}$ of elements of $\mathcal{X}$ such that the conditional probability of $X_{h+1}$ given $X_h=x$ is the uniform distribution on $\{y\in\mathcal{X}\,:\,d(x,y)=1\}$ (see \cite{GOBEL1974311} for finite graphs and \cite[I.1.C. Random walks on graphs]{woess_2000} for locally finite graphs).
    
    \smallskip
    
    The end-to-end distance, also called distance from origin, is defined as $d(X_0,X_h)$ and quantifies the remoteness to the initial value. The asymptotic behaviour of this function of the random walk is often referred to as escape rate and has been investigated in various settings and from different points of view (see \cite{gouezel:hal-01202000} for a theoretical analysis for random walks defined on hyperbolic groups, \cite[II.8.A. The rate of escape]{woess_2000} for a study in a general setting, and \cite{Mazur} in the context of a polymer molecule model).
    
    \smallskip
    
    For example, on $\mathbb{Z}^p$ with $p\geq3$, the random walk (almost surely) goes to infinity, and behaves at first order when $h$ goes to infinity as
    $$\sqrt{\frac{2h}{p}}\frac{\Gamma\left(\frac{p+1}{2}\right)}{\Gamma\left(\frac{p}{2}\right)},$$
    where
    $\Gamma(z) = \int_0^\infty x^{z-1}\exp(-x)dx$. When $p\in\{1,2\}$, the behaviour of the random walk is drastically different: the random walk is recurrent, meaning that it visits its origin location infinitely often with probability $1$. In other words, when $p<3$, the local connectivity of $\mathbb{Z}^p$ is small enough to avoid the escape of the random walk, while when $p\geq3$, the process escapes with the rate above obeying to the following expected property: the larger the dimension (and thus the connectivity), the faster the escape rate.
    
    \smallskip
    
    By definition, the isotropic random walk is a canonical object associated with the structure of the space $\mathcal{X}$, i.e. with its metric through its local connectivity. This example shows that the properties of the random walk, in particular the asymptotic properties, can be strongly related to the structure of the state space (the dimension in the case of the random walk on $\mathbb{Z}^p$ but it could be a more complex dependency). As a consequence, one can build an understanding of the local structure of a discrete metric space through the study of the behaviour of the canonical random walks on it. In this paper, we aim to develop such considerations for unordered rooted trees. It should be noted that this space has no classical algebraic structure, which considerably complicates its study.
    
    \smallskip
    
    One can distinguish two main types of metrics on the space of unordered rooted trees: edit distances and alignment distances, which form a particular case of the former (see the survey \cite{Bille2005} and the references therein). An edit distance is defined as the minimal cost of a sequence of allowed elementary edit operations that transform a tree into another (up to isomorphism). Typically, the edit operations consist in adding a node, deleting a node, and changing the label of a node.  It has been shown that the problem of computing a general edit distance between two unordered labeled trees is NP-hard (see \cite{Bille2005,zhang96}). Taking this into account, one line of research is to develop fixed-parameter polynomial algorithms, which aim to estimate a given parameter assuming that it is fixed in the evaluation of the time-complexity. For instance, \cite{AKUTSU2011352} develops and investigates an exact algorithm for a general tree edit distance that runs in polynomial time but assuming the fact that the target distance is fixed. Another option is to add restrictions to the definition of the distance so that the underlying optimization problem becomes polynomial. \cite{zhang96} has introduced a slight constraint on the set of edit operations which made reachable a polynomial algorithm.
Interestingly, it can be noted that the less restrictive conditions imposed in \cite{6786615} are only valid for ordered trees, while the problem remains NP-hard for unordered trees (see \cite[3.4 Constrained edit distance]{Bille2005}). The induced distance is referred to as Constrained Unordered Tree Edit Distance (CUTED) and is denoted by $\Dist$ in this paper. It is one of the most general  metrics known that remains polynomial without assuming that some parameters are fixed or bounded. CUTED algorithm indicates the original polynomial algorithm developed in \cite{zhang96}.

\smallskip

    The local connectivity of trees endowed with CUTED is non-trivial. First of all, there is a minimal element: the tree composed of a unique node. This tree has a unique neighbour which is composed of two vertices, one being the unique child of the other. And so on, the space of trees is arranged in successive layers: the neighbours of a tree of size $n$ have $n-1$ or $n+1$ nodes, and all the trees of size $n+1$ can be reached from trees of size $n$. There is no closed-form expression for the number of trees of size $n$ (EIS A000081) but it grows fastly as $\lambda\beta^n/n^{3/2}$ with $\lambda\simeq0.43992$ and $\beta\simeq2.48325$ \cite{FlajoletSedgewick}. In addition, from a tree with $n$ nodes, there are from $n$ to $2n$ ways to construct a bigger tree that remains at distance $1$, while the number of nodes that can be deleted in CUTED is always less than $n$. However, some of these editions result in the same tree. The local connectivity is hard to describe because it is highly dependent on the considered tree, even at fixed size: the fishbone tree with $n$ nodes has $\Theta(3n/2)$ neighbours of size $n+1$ while the rake tree with $n$ nodes has only $2$ neighbours of size $n+1$. This draws a complex funnel-like structure, which first 5 layers are represented in Fig.\,\ref{fig:ex:walk}. This being established and recalling the asymptotic behaviour of random walks on $\mathbb{Z}^p$ for which the number of neighbours of any state is $2p$, one can expect a random walk on trees to escape fastly from its starting point. Nevertheless, the form of the escape rate function is far from obvious. In this paper, we investigate, from a numerical point of view, the escape rate of the canonical random walk on trees endowed with CUTED.
    
    \smallskip
    
    A natural way to proceed is to generate numerically a large number $N$ of random walks within a large time window $[0,H]$ in order to estimate accurately the average escape rate, then evaluate its dependency on the number of elementary moves. In other words, we need to simulate $N$ independent random walks on trees $(X_h^i)_{0\leq h\leq H}$, $1\leq i\leq N$, and keep track of the end-to-end distance processes $(\Dist(X_0^i,X_h^i))_{0\leq h\leq H}$. Consequently, $N\times H$ edit distances need to be evaluated. In light of the quadratic complexity of CUTED algorithm \cite{zhang96}, this approach clearly requires intensive computing resources, making it impractical as soon as the trees get large.
    
    \smallskip
    
    In the present paper, we remark that $X_{h+1}^i$ is not completely independent of $X_{h}^i$ but obtained from it by one of the elementary moves allowed in CUTED. In this context, we show that the edit distance $\Dist(X_0^i,X_{h+1}^i)$ does not require to process CUTED algorithm in full, but can be evaluated fastly from $\Dist(X_0^i,X_h^i)$ (and the computations made to evaluate it) and the knowledge of the edit operation from $X_h^i$ to $X_{h+1}^i$. More precisely, we exhibit an incremental algorithm for CUTED, which achieves much shorter computation times than the original version: for instance, the computation time is divided by almost $100$ for trees of size $500$. Applying this idea from the initial value of the random walk, one can keep track of the distance process without any intensive computations. This notably makes possible to estimate accurately the escape rate function of the canonical random walk on trees on a commercial laptop ($90\text{M}$ distances on trees up to several hundreds of nodes were evaluated to write this paper).
    
    The paper is organized as follows. Section~\ref{s:rw} is devoted to precise definitions of the concepts of interest: trees in Subsection~\ref{ss:def:trees}, CUTED in Subsection~\ref{ss:zhang:operations}, and balanced random walk on trees in Subsection~\ref{ss:isorw}. Our incremental algorithm for CUTED is developed in Section~\ref{s:effcomp}, while its pseudo-code is provided in Appendix~\ref{pseudo-alg}. A simulation study is presented in Section~\ref{s:simu}: improvement in computation times are given in Subsection~\ref{ss:comptimes} and Subsection~\ref{ss:sharpesc} presents our investigations on the sharp estimation of the escape rate of random walks on trees. Other aspects of the work that can be of interest are discussed in Section~\ref{s:discussion}.

\section{Random walk on trees endowed with CUTED}
\label{s:rw}

\subsection{Definition of unordered trees}
\label{ss:def:trees}

\paragraph{Rooted trees} A rooted tree $T$ is a connected acyclic digraph such that there is a unique vertex $\Root(T)$ (the root of $T$) that has no parent, and any vertex different from the root has exactly one parent. The parent of a vertex $w$ is denoted $\Parent(w)$ and the set of all its ancestors by $\Parent^+(w)$, while all the nodes that have $v$ for parent form the set of children of $v$, denoted $\Children(v)$. $\Leaves(T)$ is the set of leaves of $T$ that are all the vertices without children. The tree without any vertex represents the empty tree and is denoted $\Theta$. The subtree $T[v]$ rooted in $v$ is the subgraph of $T$ composed of $v$ and all its descendants in $T$ with edges inheriting from $T$. $F[v]$ denotes the forest of subtrees emanating from $v$, i.e. the set of trees $T[w]$ where $w\in\Children(v)$.

\paragraph{Unordered trees} In this paper, we consider unordered rooted trees for which the left-to-right order among sibling vertices is not significant. Unordered rooted trees, simply called unordered trees or trees in the sequel, are defined from the definition of tree isomorphism. A one-to-one correspondence $\Phi:T\to T'$ is called a tree isomorphism if $w$ is a child of $v$ in $T$ implies that $\Phi(w)$ is a child of $\Phi(v)$ in $T'$. $T$ and $T'$ are said isomorphic whenever there exists a tree isomorphism between them.
The existence of a tree isomorphism defines an equivalence relation on the set of rooted trees. The class of unordered rooted trees is the quotient set of rooted trees by the existence of a tree isomorphism. Remarkably, one can determine if two trees of size $n$ are isomorphic, which the most elementary operation required to handle unordered trees, in $O(n)$ \cite[Example 3.2 and Theorem 3.3]{Aho:1974:DAC:578775}.

\paragraph{Characteristics of trees} A tree $T$ can be described by several integer-valued characteristics that will be used in this paper, namely:
\begin{itemize}
    \item its size $\Size T$ that counts the number of vertices of $T$;
    \item its height $\Height(T)$ that is the length of the longest path from the root to the leaves;
    \item its outdegree $\Outdegree(T)$ that is the maximum number of children that can be found in $T$;
    \item its Strahler number, recursively defined on vertices $v$ of $T$ as $\Strahler(v) = 1$ if $v\in\Leaves(T)$, and
    $$ \Strahler(v) = \max_{c\in\Children(v)} \Strahler(c) +
    \left\{
    \begin{array}{cll}
    0 & \text{if} & \Size[\argmax_{c\in\Children(v)} \Strahler(c)] = 1,\\
    1 & \text{if} & \Size[\argmax_{c\in\Children(v)} \Strahler(c)] \geq2,
    \end{array}
    \right.
    $$
    else, the Strahler number of $T$ being defined as the Strahler number of its root.
\end{itemize}

\subsection{Definition of CUTED from edit operations} \label{ss:zhang:operations}

Editing distances consist in evaluating the number of elementary operations transforming one tree in another. General editing operations for unordered trees were introduced in \cite{10.1016/0020-0190(92)90136-J}, but this paper also shows that computing the induced distance is NP-complete \cite[Theorem 11]{10.1016/0020-0190(92)90136-J}. In \cite{zhang96}, the set of editing operations is slightly constrained so that a polynomial algorithm for computing the induced distance, refered to as CUTED in this paper, is exhibited. CUTED is defined from the following set of edit operations of a tree $T$.
\begin{itemize}
    \item Add leaf under $u$: add a new vertex $k$ in the set of children of $u$ such that $k$ is a leaf.
    \item Del leaf $k$: remove leaf $k$ from the set of children of its parent.
    \item Add internal node under $u$: replace the set of children of $u$ by a unique vertex $k$ so that the set of children of $k$ is the previous set of children of $u$.
    \item Del internal node $k$ (only if $k$ is the unique child of its parent $u$): remove $k$ so that the set of children of $k$ becomes the children of $u$.
\end{itemize}
For the sake of clarity, these 4 operations are presented in Fig.\,\ref{fig:editing:operations}.

\begin{figure}[h]
\centering
\subcaptionbox{\label{fig:editing:operations-a}}{\includegraphics[width=4.5cm]{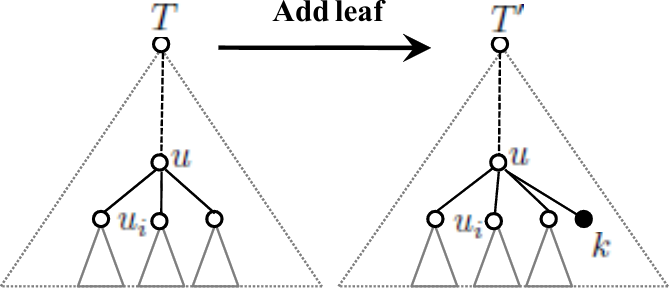}}
\hspace{2cm}\subcaptionbox{\label{fig:editing:operations-b}}{\includegraphics[width=4.5cm]{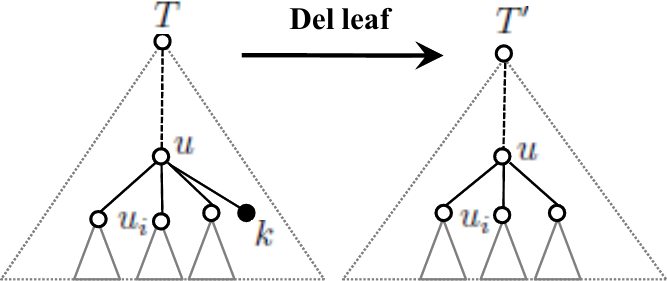}}\\
\subcaptionbox{\label{fig:editing:operations-c}}{\includegraphics[width=4.5cm]{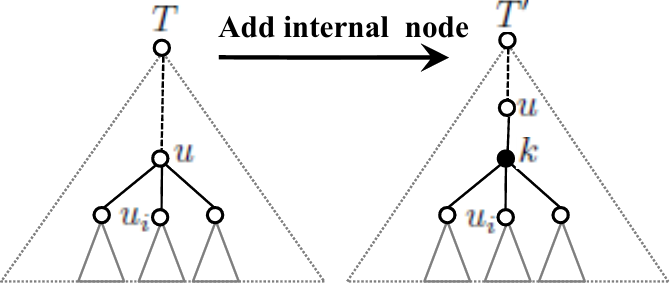}}
\hspace{2cm}\subcaptionbox{\label{fig:editing:operations-d}}{\includegraphics[width=4.5cm]{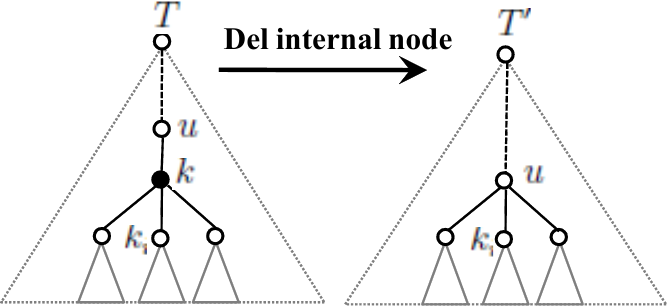}}
\caption{Elementary tree editing operations transforming a tree $T$ into an other tree $T'$, in such edition either:  (a) a new leaf $k$ is inserted under vertex $u$, (b) leaf $k$ of a parent $u$ in $T$ is deleted from $T'$, (c) a new vertex $k$ is added under the vertex $u$  such that $k$ becomes the unique child of $u$ and the children set of $k$ is the previous $u$ children set, (d) $k$ is the unique child of its parent $u$ then remove $k$ so that the  $k$ children set becomes the children set of $u$.
\label{fig:editing:operations}}
\end{figure}

The tree obtained after an edit operation $o$ applied to a tree $T$ is denoted $o(T)$. An edit operation sequence is an ordered list of edit operations $s = (o_n,\dots,o_1)$ applicable to $T$, i.e. such that, for any $1\leq i\leq n-1$, operation $o_{i+1}$ makes sense to be applied to $o_i\circ\dots\circ o_1(T)$.

\smallskip

The edit distance between two trees $T$ and $T'$ is the minimal length of edit operation sequences transforming $T$ into a tree isomorphic to $T'$, i.e.
\begin{equation}\label{eq:def:tree dist}
 \Dist(T,T') = \min_{s(T)=T'} \Size s .
\end{equation}
$\Dist$ defines a distance metric on the space of unordered trees \cite[2.1.\,Editing operations and editing distance between unordered labeled trees]{10.1016/0020-0190(92)90136-J}.

\subsection{Balanced random walk on tree space}
\label{ss:isorw}

A Markov chain on a discrete space $\mathcal{X}$ is defined from a transition kernel $Q$, that gives, for any $(x,y)\in\mathcal{X}^2$, the probability $Q(x,y)$ to go from $x$ to $y$ in one step. If the space $\mathcal{X}$ is augmented with a graph structure, i.e. if a set $\mathcal{E}$ of (undirected) edges on $\mathcal{X}$ is given, a random walk is defined as a Markov chain such that $Q(x,y)>0$ if and only if $(x,y)\in\mathcal{E}$, i.e. if $y\in\Neighbors(x)$ where $\Neighbors(x)$ is the set of neighbors of $x$,
$$\Neighbors(x) = \{y\in\mathcal{X}~:~(x,y)\in\mathcal{E}\}.$$
(see \cite[I.1.C. Random walks on graphs]{woess_2000} where these random walks are referred to as of nearest neighbor type). The transitions of an isotropic random walk, also called simple random walk, should not privilege one specific direction: in general, one expects $Q(x,\cdot)$ to be the uniform distribution on the set $\Neighbors(x)$ (supposed to be finite for any $x$), i.e.
$$Q(x,y) =
\left\{
\begin{array}{cl}
\frac{1}{\Size\Neighbors(x)} & \text{if}~y\in\Neighbors(x),\\
0 & \text{else,}
\end{array}
\right.
$$
(see \cite{GOBEL1974311} for finite graphs and \cite[eq.\,(1.19)]{woess_2000} for locally finite graphs).

\smallskip

CUTED induces a natural locally finite graph structure on the set of unordered trees, for which two trees form an edge if they are at distance $1$. In this case, the set $\Neighbors(T)$ of neighbors of $T$ is given by
\begin{eqnarray*}
\Neighbors(T)&=& \{T'~:~\Dist(T,T') = 1\},\\
&=& \{o(T)~:~o\in\mathcal{O}\},
\end{eqnarray*}
where $\mathcal{O}$ denotes the set of all possible operations applicable to $T$. This set can be decomposed as
$$\mathcal{O} = \mathcal{O}_{\text{AL}} \cup \mathcal{O}_{\text{DL}} \cup \mathcal{O}_{\text{AIN}} \cup \mathcal{O}_{\text{DIN}},$$
where
\begin{itemize}
    \item $\mathcal{O}_{\text{AL}}$ stands for the set of add leaf operations, which cardinality is $\Size T$;
    \item $\mathcal{O}_{\text{DL}}$ stands for the set of del leaf operations, which cardinality is $\Size\Leaves(T)$;
    \item $\mathcal{O}_{\text{AIN}}$ stands for the set of add internal node operations, which cardinality is $\Size T-\Size\Leaves(T)$;
    \item $\mathcal{O}_{\text{DIN}}$ stands for the set of del internal node operations, which cardinality is the number of single children that are not leaves, which is bounded by $\Size T-\Size\Leaves(T)-1$.
\end{itemize}
Consequently, the number of operations that make the tree increase, $2\Size T-\Size\Leaves(T)$, is much larger than the number of operations that make the tree decrease, roughly upper-bounded by $\Size T-1$. Thus, choosing a random neighbor of $T$ by selecting an edit operation with uniform distribution on $\mathcal{O}$ introduces a bias that tends to augment the size of the tree, which is not the expected behavior of a balanced random walk.

\smallskip

In order to avoid this bias, we propose to select with probability $1/2$ an adding operation, chosen with uniform distribution on $\mathcal{O}_{\text{AL}}\cup\mathcal{O}_{\text{AIN}}$, and with probability $1/2$ a deleting operation, chosen with uniform distribution on $\mathcal{O}_{\text{DL}}\cup\mathcal{O}_{\text{DIN}}$, which leads to
\begin{equation}\label{eq:def:Q}
Q(T,T') =
\left\{
\begin{array}{cl}
\frac{1}{2(\Size\mathcal{O}_\text{AL}+\Size\mathcal{O}_\text{AIN})} & \text{if}~T'=o(T)~\text{with}~o\in\mathcal{O}_\text{AL}\cup\mathcal{O}_\text{AIN} ,\\
\frac{1}{2(\Size\mathcal{O}_\text{DL}+\Size\mathcal{O}_\text{DIN})} &\text{if}~T'=o(T)~\text{with}~o\in\mathcal{O}_\text{DL}\cup\mathcal{O}_\text{DIN} ,\\
0 & \text{else.}
\end{array}
\right.
\end{equation}
In the sequel, $(T_h)_{h\geq0}$ denotes the balanced random walk defined from the transition kernel $Q$ given in \eqref{eq:def:Q}, i.e.
$$\forall\,h\geq0,~\mathbb{P}(T_{h+1} = y | T_h = x) = Q(x,y).$$
The first layers of the state space as well as a sample trajectory are presented in Fig.\,\ref{fig:ex:walk}. The isotropic random walk, that assigns the same probability to all the neighbours, is discussed in Subsection~\ref{ss:isotropic}.

\begin{figure}[t]
\centering\includegraphics[scale=0.65]{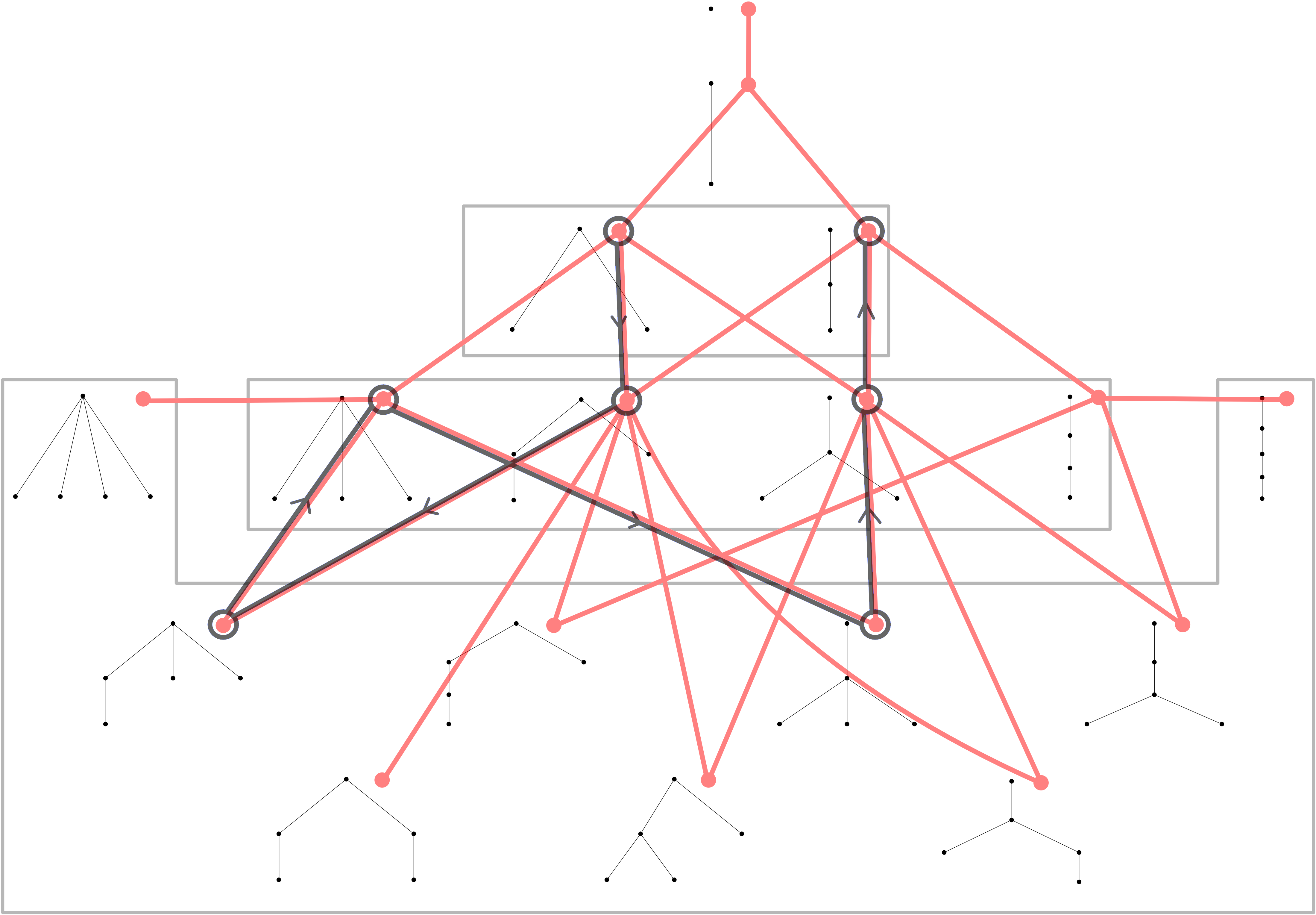}
\caption{Graph structure induced by CUTED on the $17$ unordered rooted trees with at most $5$ nodes and the first steps of a sample random walk on it starting from a tree with $3$ nodes.}
\label{fig:ex:walk}
\end{figure}

\smallskip

The behavior of $(T_h)_{h\geq0}$ can be described through integer-valued characteristics of trees, namely the size $(\Size T_h)_{h\geq0}$, the height $(\Height(T_h))_{h\geq0}$, the outdegree $(\Outdegree(T_h))_{h\geq0}$, or the Strahler number $(\Strahler(T_h))_{h\geq0}$ (see Subsection~\ref{ss:def:trees} for definition of these quantities). However, one important feature of random walks is the end-to-end distance (see \cite{gouezel:hal-01202000,Mazur,woess_2000} for various contexts and objectives), which evaluates the remoteness to the initial value, given in the framework of CUTED by $(\Dist(T_0,T_h))_{h\geq0}$.

\section{Efficient computation of the tree edit distance in a random walk context}
\label{s:effcomp}

During a random walk, our goal is to analyze the behaviour of the above tree-edit distance $\Dist (T_{0},T_h)$ given in \eqref{eq:def:tree dist} between the origin tree and the current tree. For unordered trees, this computation can be carried out using the \emph{dynamic programming} principle in a recursive and bottom-up manner, from the leaves of the trees to their roots \cite{zhang96}. In this section, we will show how this computation can be carried out efficiently, in an incremental manner during the random walk.

\subsection{Recursive computation of tree edit distance}

\cite{zhang96} introduced an efficient algorithm to compute edit distances between unordered trees (here corresponding to $T_0$ and $T_h$), assuming simple edit operations (see Subsection~\ref{ss:zhang:operations}). This algorithm relies on the following set of recursive equations that proceed bottom up from the leaves to the tree root,
\begin{equation}
\begin{aligned}
\Dist(v,w) &=\min\begin{cases}
\Dist(\Theta,w)+\underset{w_{i}\in \Children(w)}{\min}\left\{ \Dist(v,w_{i})-\Dist(\Theta,w_{i})\right\},\\
\Dist(v,\Theta)+\underset{v_{i}\in \Children(v)}{\min}\left\{ \Dist(v_{i},w)-\Dist(v_{i},\Theta)\right\},\\
d(v,w)+\Dist_F(v,w),
\end{cases} \\
\Dist_F(v,w) &=\min\begin{cases}
\Dist_F(\Theta,w)+\underset{w_{i}\in \Children(w)}{\min}\left\{ \Dist_F(v,w_{i})-\Dist_F(\Theta,w_{i})\right\},\\
\Dist_F(v,\Theta)+\underset{v_{i}\in \Children(v)}{\min}\left\{ \Dist_F(v_{i},w)-\Dist_F(v_{i},\Theta)\right\},\\
\Dist_{FM}(v,w)=\underset{(v_{i},w_{j})\in M^\ast(v,w)}{\sum}\Dist(v_{i},w_{j}),
\end{cases}
\end{aligned}
\label{recursive-equations}
\end{equation}
where $\Dist(v,w)$ and $\Dist_F(v,w)$
respectively represent the distance between subtrees $T_0[v]$ and $T_h[w]$ and the distance between the forests $F_{0}[v]$ and $F_h[w]$ (see Fig.\,\ref{fig:trees T1 and T2 with partial solutions}). $\Dist_{FM}(v,w)$ is the cost of $M^\ast(v,w)\subseteq \Children(v) \times \Children(w)$ the \emph{optimal forest mapping} rooted in $v$ and $w$. For this define $\mathcal{M}(v,w)$ the set of all mappings between $F[v]$ and $F[w]$. Based on the above edit operations, the only valid mappings between forests must verify \cite{zhang96},
$$
    \forall\,(x,y)\in M(v,w),~\forall\,(x',y')\in M(v,w),~x=x' \,\Leftrightarrow\,y=y'.
$$

\noindent
An \emph{optimal forest mapping} $M^{*}$ between $F_0[v]$ and $F_h[w]$ must satisfy,
$$
    M^{\ast}(v,w)=
    \argmin_{M(v,w) \in \mathcal{M}(v,w)} \left( \sum_{(x,y)\in M(v,w)} \Dist(x,y)+\sum_{x\in M^{-}(v,w)}\Dist(x,\Theta)+\sum_{y\in M^{+}(v,w)}\Dist(\Theta,y)\right).
$$
$M^{-}(v,w)$ and $M^{+}(v,w)$ represent respectively the subsets of $\Children(v)$ and $\Children(w)$ which are not included in the mapping $M(v,w)$.

\begin{figure}[t]
\centering\includegraphics[scale=0.7]{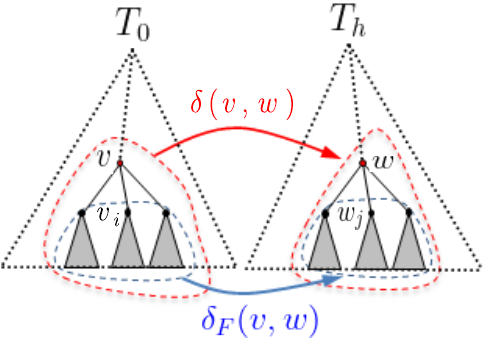}
\caption{Intermediate distances involved in computation of the mapping distance of trees $T_{0}$ and $T_h$. Dashed lines surround the elements to be mapped for each pair $(v,w)$: $\Dist(v,w)$ is the mapping distance between the subtrees $T_0[v]$ and $T_h[w]$, and $\Dist_F(v,w)$ is the mapping forest distance between $F_0[v]$ and $F_h[w]$.
\label{fig:trees T1 and T2 with partial solutions}}
\end{figure}

\smallskip

\cite{zhang96} modeled this problem as a minimum cost maximum bipartite matching problem. A weighted bipartite graph $G$ is constructed, made of sets of vertices: a first set $V_0$ represents the trees of the forest $F_0[v]$ and a second set $V_h$ represents the trees of $F_h[w]$. In addition, if the two forests have different sizes, a vertex $\theta$, representing the empty tree, is added to the set with lower number of vertices. An edge is created between each pair of vertices in $V_0 \times V_h$, and is attached a weight $\gamma_{ij}$ corresponding to the tree edit distance between the subtrees $T_0[v_i]$ and $T_h[w_j]$ represented by the two connected vertices in $G$. Two nodes representing the source $s$ and the sink $t$ are added to $G$ (see Fig.\,\ref{fig:matching graphs}). Using this graph, the above forest mapping problem is solved as the \emph{minimum cost maximum flow} (MCF) problem on $G$ \cite{kiraly-kovacs2012}. 

\begin{figure}[th]
\centering\includegraphics[scale=0.7]{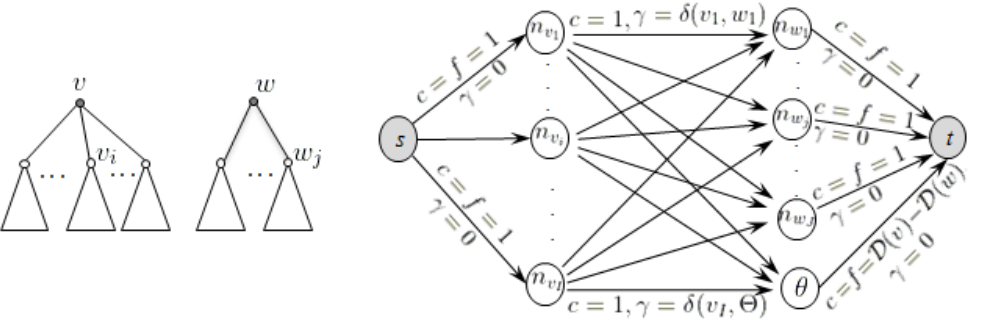}
\caption{Two forests $F_0[v]$ and $F_h[w]$ (left) of mapping cost computed from the associated matching graph (right), where for each edge the triplet $(f,c,\gamma)$ denotes the flow, the capacity and the cost. Note that for this example $\Outdegree(v)>\Outdegree(w)$.
\label{fig:matching graphs}}
\end{figure}

\paragraph{MCF problem \cite{kiraly-kovacs2012}}

Let $G=(V,E)$ be a weighted connected directed graph. Each edge $(i,j) \in E$ is associated with 2 non-negative integers $f_{ij}$ and $c_{ij}$ representing respectively the edge \textit{flow} and \textit{capacity}, and a real $\gamma_{ij}$ defining the \textit{cost} per unit flow $f_{ij}$ on the arc. The \emph{cost of a flow} $f$ on $G$ is defined as
$$\gamma(f)=\sum_{(i,j)\in E}\gamma_{ij}f_{ij}.$$ 

In addition, we assume that there exist two nodes in $G$, denoted $s$ and $t$, respectively called the source and target nodes such that $s$ has no entering edge, and $t$ has no outgoing edge. All the other nodes are called \textit{transit} nodes.

\smallskip

A flow $f$ is \emph{feasible}, if and only if it meets the following constraints: 
\begin{enumerate}
\item Capacity: for any $(i,j)\in E$, $0\leq f_{ij}\leq c_{ij}$.
\item Flow conservation: for any transit node $i$, $\underset{(i,j)\in E}{\sum}f_{ij}=\underset{(j,i)\in E}{\sum}f_{ji}$.
\item Supply-demand conservation: the source supply value $\underset{(s,j)\in E}{\sum}f_{sj}$ is equal to the sink demand value $\underset{(i,t)\in E}{\sum}f_{it}$.
\end{enumerate}

A Minimum Cost Maximum Flow (MCF) problem consists in finding an \textit{optimal} flow $f^*$ representing the feasible maximal flow on the graph which minimizes the flow cost.

\smallskip

The search of the optimal forest mapping from $F_{0}[v]$ to $F_h[w]$ thus corresponds to solving the MCF problem associated with the graph $G$ encoding forests $F_{0}[v]$ and $F_h[w]$. 
Different algorithms can be used for the MCF resolution, such as \cite{Edmons-Karp72,Edmons-karp2003},  Network Simplex \cite{kiraly-kovacs2012, ahuja2017}, and the Cost-Scaling algorithm \cite{Tarjan1990}. In \cite{zhang96}, the former algorithm, improved by Tarjan \cite{tarjan1983}, is used, which induces the final time complexity of the overall algorithm $O(\Size T_{0} \times \Size T_{h} \times (\Outdegree(T_{0})+ \Outdegree(T_{h}))\times \log_{2}(\Outdegree(T_{0})+ \Outdegree(T_{h}))$.

\subsection{Incremental tree edit distance}
\label{ss:inczhang}

The above set of recursive equations \eqref{recursive-equations} makes it possible to compute the tree edit distance between 2 trees in a time essentially polynomial in the product of the input tree sizes. However, during a random walk $(T_{h})_{h\geq 0}$, performing this computation at each step between the current tree $T_h$ and the original tree $T_{0}$ would be prohibitive. In this section, we show that this computational cost can be drastically reduced by making use of the fact that each tree $T_h$ is followed by a tree $T_{h+1}$ that results from an elementary edit operation made on $T_h$. 

\smallskip

As we have seen above, such an elementary operation can either be a vertex insertion or deletion in the tree $T_h$ (see Fig.\,\ref{fig:editing:operations} in which $T$ and $T'$ are here respectively referenced as trees $T_h$ and $T_{h+1}$). We remark that because the tree $T_{h+1}$ differs only by one vertex from the tree $T_h$, the recursive computation of $\Dist(T_0,T_{h+1})$ will differ only slightly from the recursive computation of $\Dist(T_0,T_{h})$: most of the intermediate distances computed during the recursion are actually the same for a majority of subtrees not affected by the edit operation. In the sequel, we make use of this property to reduce the complexity of the computation of $\Dist(T_0,T_{h+1})$, by updating in \eqref{recursive-equations} only the distances of $\Dist(T_0,T_{h})$ impacted by the edit operation at step $h+1$.

\smallskip

In the computation of the distance between two trees, we note that a number of intermediate quantities, $\Dist(v,w)$, $\Dist(v,\Theta)$, $\Dist(\Theta,w)$, $\Dist_F(v,w)$, $\Dist_F(v,\Theta)$, $\Dist_F(\Theta,w)$ and $M^\ast(v,w)$, are recursively computed for each pair $(v,w)$ during the algorithm following \eqref{recursive-equations}. These quantities define the \textit{distance information at} $(v,w)$.

\smallskip

Let us assume that we already computed the distance between $T_0$ and $T_{h}$ using \eqref{recursive-equations} and that, at the next step, a vertex $k$ is either inserted in or deleted from the tree $T_h$ to produce $T_{h+1}$. 
Then, in the recursion of \eqref{recursive-equations} not all the intermediate quantities need be recomputed to compute $\Dist(T_0,T_{h+1})$. Indeed, all the distances $\Dist(v,w)$ for $w$ not an ancestor of $k$, are left unchanged by the edit operation changing $T_h$ in $T_{h+1}$. Only the distances $\Dist(v,w)$ for $w\in \Parent^+(k)$ actually need to be re-evaluated (see Fig.\,\ref{fig:mod-cases}).

\begin{figure}[th]
\centering\includegraphics[scale=0.5]{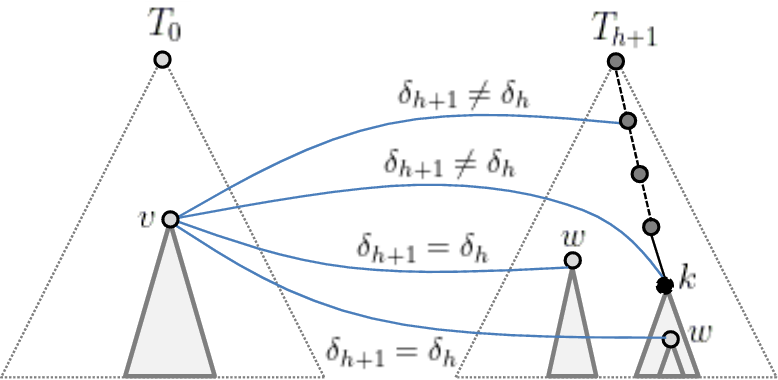}
\caption{Tree $T_{h+1}$ in which a node $k$ is edited (just inserted or about to be deleted). In step $h+1$, the only distance informations to be updated are for the vertex $k$ and for all $k$ ancestors, while for all remaining $T_{h+1}$ vertices the distance information is preserved compared with that computed in step $h$.
\label{fig:mod-cases}}
\end{figure}

\smallskip

How the distance information at $(v,w)$ used to compute $\Dist(T_0,T_{h})$ must be updated at step $h+1$ to compute $\Dist(T_0,T_{h+1})$ depends on the type of edit operation (see Fig.\,\ref{fig:editing:operations}) and on the position of the vertex $w$ with respect to the edited vertex $k$ (see Fig.\,\ref{fig:mod-cases}). Let us consider the different situations leading to different updating strategies.

\paragraph{Leaf insertion (see Fig.\,\ref{fig:editing:operations-a})}
$T_{h+1}$ results from the insertion of a leaf $k$ in $T_{h}$. Then, for all $v$ in $T_{0}$ and for every vertex $w$ in $T_{h+1}$, most of the distance information at $(v,w)$ remains unchanged compared with what was computed at the previous step, except in three cases:
    \begin{enumerate}[label=\roman*.]
    \item $w=k$: the distance information at $(v,w)$ is set to that of $(v,l)$, where $l$ is any other leaf in $T_{h}$.  
    
   \item $w=\Parent(k)$: at step $h+1$, the forest $F_{h+1}[w]$ is made of the forest $F_h[w]$  augmented with the subtree $T_{h+1}[k]$ (here reduced to only one leaf vertex). This means that the distance information at $(v,w)$ at step $h+1$ must be computed using \eqref{recursive-equations} to account for the fact that (see Fig.\,\ref{fig:changes-all-cases-a}):  
        \begin{itemize}
        \item The distance information at $(v,k)$ is new.
        
        \item For every $v_i\in \Children(v)$, the distance informations at $(v_i,w)$ and at $(\Theta,w)$ have potentially changed due to the leaf insertion.
        
        \item The optimal forest mapping $M^\ast(v,w)$ obtained in step $h$ is no longer valid at step $h+1$. We have to integrate the new subtree $T_{h+1}[k]$ in the mapping. See problem B in paragraph \textit{Re-evaluation of forest mappings} below.
        \end{itemize}

    \item $w$ is an ancestor of $\Parent (k)$: at step $h+1$, the forest $F_h[w]$ undergoes an edition without affecting the $w$ children list. Therefore, the distance information at $(v,w)$ at step $h+1$ must be recomputed using \eqref{recursive-equations} to account for the following changes with respect to the computation at step $h$ (see Fig.\,\ref{fig:changes-all-cases-b}):
    
        \begin{itemize}
        \item The distance information at $(v_i,w)$ for every $v_i\in \Children(v)$, the distance information at $(v,w_j)$ for $w_j\in\Parent^+(k)$, and the distance information at $(\Theta,w)$ have potentially changed due to the insertion.
    
        \item Here the forest mapping at step $h$ remains a valid mapping at step $h+1$ but should be further optimized to account for potential changes in the distance between subtrees. See problem A in paragraph \textit{Re-evaluation of forest mappings} below.
        \end{itemize} 
    
    \end{enumerate}

\paragraph{Leaf deletion (see Fig.\,\ref{fig:editing:operations-b})}
The leaf $k$ has just been deleted from $T_{h}$, and the only tree distance informations that are affected by this deletion are for the two cases:

    \begin{enumerate}[label=\roman*.]
    \item $w=\Parent(k)$: at step $h+1$ the subtree $T_{h}[k]$ (here reduced to only one leaf vertex) is substracted from the forest $F_h[w]$. The distance informations at $(v,w)$ in step $h+1$ must be recomputed using \eqref{recursive-equations} to account for the following changes with respect to the distance computation in step $h$ (see Fig.\,\ref{fig:changes-all-cases-c}):
        \begin{itemize}
        
        \item The previous distance informations at $(v,k)$ are excluded from the computation.
        
        \item For every $v_i\in \Children (v)$, the distance informations at $(v_i,w)$ and at $(\Theta,w)$ have potentially changed due to the deletion.
    
        \item The optimal forest mapping $M^\ast(v,w)$ computed in step $h$ is no longer valid at step $h+1$, since the subtree $T_{h}[k]$ has to be excluded from the mapping. See problem B in paragraph \textit{Re-evaluation of forest mappings} below.
        \end{itemize}

    \item $w$ is an ancestor of $\Parent (k)$: the distance information needs to be updated as in (iii) of leaf insertion (see above).
    \end{enumerate}

\paragraph{Internal node insertion (see Fig.\,\ref{fig:editing:operations-c})}
$T_{h+1}$ results from the insertion of an internal vertex  $k$ in $T_{h}$. Then, the only tree distance informations that are affected by this insertion are as follows:

    \begin{enumerate}[label=\roman*.]
    \item $w=k$: let us denote $u$ the vertex of $T_{h}$ under which the insertion takes place (see Fig.\,\ref{fig:editing:operations-c}). In the new tree $T_{h+1}$, the insertion of an internal vertex $w$ under $u$ means that, by definition, the forests $F_h[u]$ and $F_{h+1}[w]$ are identical. As a consequence, the distance information at step $h+1$ at $(v,w)$ is equal to that previously obtained for $(v,u)$ at step $h$ and does not need to be re-evaluated. 
    
    \item $w=\Parent(k)$: if $k$ is an internal vertex inserted under $w$ in $T_{h+1}$, a first step of the edition at level of vertex $w$ is to empty the forest $F_{h+1}[w]$ before adding to it the subtree $T_{h+1}[k]$. Meaning that, the distance information at $(v,w)$ is at first set to that of $(v,l)$, $l$ any leaf in $T_h$. Next, to deal with the insertion of the subtree $T_{h+1}[k]$ in the temporarily empty forest $F_{h+1}[w]$, the obtained distance information at $(v,w)$ is re-evaluated in a similar way as for case (ii) of leaf insertion (see above).
    
    \item $w$ is an ancestor of $\Parent (k)$: see above leaf insertion (iii).
    \end{enumerate} 
 
\paragraph{Internal node deletion (see Fig.\,\ref{fig:editing:operations-d})}

Internal vertex $k$ has just been deleted from $T_{h}$, and the only tree distance informations that are affected by this deletion are:

    \begin{enumerate}[label=\roman*.]
    \item $w=\Parent(k)$: if $k$ is an internal vertex deleted under $w$, then the forests $F_{h+1}[w]$ and $F_h[k]$ are identical, and the distance information at $(v,w)$ in step $h+1$ is equal to the distance information at $(v,k)$ of step $h$ and does not need to be re-evaluated.
    
    \item $w$ is an ancestor of $\Parent (k)$: see above leaf insertion (iii).
    \end{enumerate}     
    
\begin{figure}[th]
\centering
\subcaptionbox{\label{fig:changes-all-cases-a}}{\includegraphics[scale=0.4]{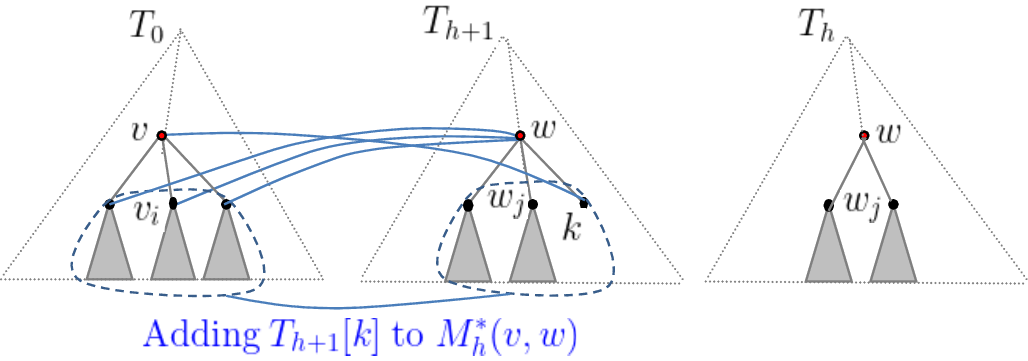}}\hspace{1.2cm}
\subcaptionbox{\label{fig:changes-all-cases-b}}{\includegraphics[scale=0.4]{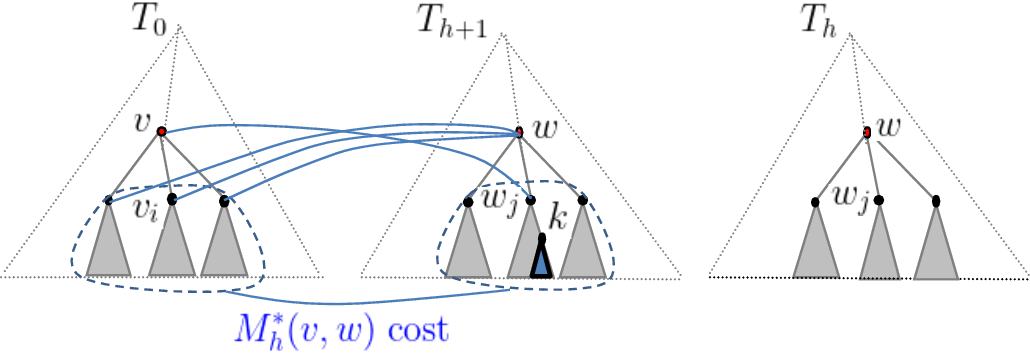}}\\
\subcaptionbox{\label{fig:changes-all-cases-c}}{\includegraphics[scale=0.4]{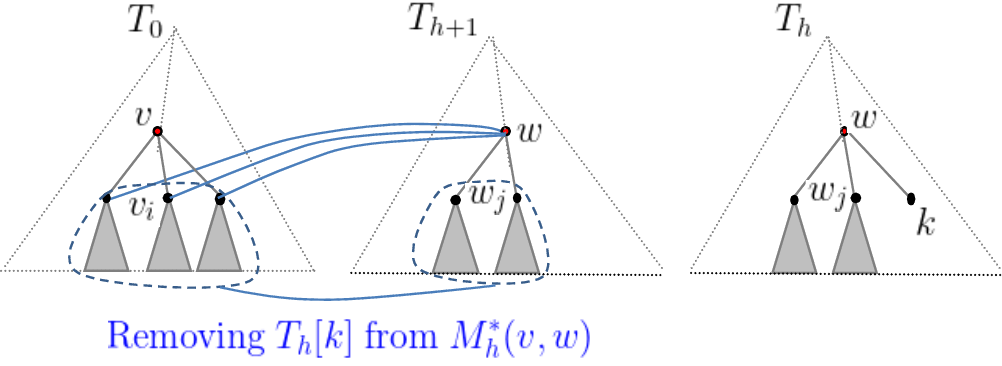}}\hspace{1.2cm}
\phantom{\includegraphics[scale=0.4]{fig6b.pdf}}
\caption{The intermediate distance information -at connected vertices pairs- which are included in $\Dist(v,w)$ computation and changing between step $h$ and step $h+1$ after a vertex $k$ edition. The variation of the $(v,w)$ optimal forest mapping is depicted as follows:
(a) After a leaf $k$ insertion under vertex $w$, $T_{h+1}[k]$ has to be integrated in the previous mapping $M_h^\ast(v,w)$. (b) For any editing operation on $k$, if $w$ is an ancestor of $\Parent(k)$ then the previous mapping $M_h^\ast(v,w)$ has to be updated according to the new costs $\Dist(v_i,w_j)$ for $w_j$ an ancestor of $k$ and for any $v_i$ a child of $v$.
(c) After a leaf $k$ deletion under vertex $w$, the previous mapping $M_h^*(v,w)$ has to be updated by excluding the subtree $T_h[k]$.
 \label{fig:changes-all-cases}}
\end{figure}

\paragraph{Re-evaluation of forest mappings}
One common problem to the 4 different cases enumerated above is linked with the re-evaluation of the forest mapping $M^\ast$ at $h+1$ knowing the forest mapping at $h$. A careful analysis reveals that again, not all the computation needs be redone at step $h+1$. Two main cases arise, depending on whether the edition operation impacts the validity of the forest mapping, i.e. when $w=\Parent(k)$) at step $h$, or not, i.e. when $w$ is an ancestor of $\Parent(k)$). This leads us to solve two different problems:

\begin{enumerate}[label=\Alph*.]
    \item If the mapping at step $h$ remains valid at step $h+1$, it may no longer be optimal. A new optimal valid forest mapping $M^\ast(v,w)$ needs to be computed;
    \item If not, we need to find an initial valid mapping. We will show that such a new valid mapping can efficiently be obtained from the optimal mapping at step $h$, and solve the optimal forest mapping, starting from this updated initial valid mapping.
\end{enumerate}

The search of an optimal forest mapping can be solved as a MCF problem on a bipartite matching graph $G$ representing the two forests \cite{zhang96}. Let us study how the MCF problem can be solved in situation A, then B.

\paragraph{Resolution of problem A} 
In this case, we already know a feasible flow $f_h^*$ on $G$ that corresponds to the optimal solution of the MCF problem at the previous step. However, this flow may have now lost its optimal character due to the cost changes in $G$. Given the limited change in the graph $G$ between two consecutive steps, an obvious heuristics to optimize the search for a new optimal flow $f_{h+1}^*$, is to start the optimization from the previous optimal flow, that is feasible. For this, we chose to use the Network Simplex (NS) algorithm \cite{kiraly-kovacs2012, ahuja2017}, that can be initialized by an arbitrary feasible flow, here chosen as $f_h^*$ at step $h+1$.  

\paragraph{Resolution of problem B} 
Here, edges costs are not changed in the graph $G$, but due to the edit operation (insertion or deletion of node $k$), the right-hand side of the graph $G$ is modified (see Fig.\,\ref{fig:flow path}). This may alter the feasibility of the previous flow $f^*$ as well as its optimality. We must thus compute a new feasible flow (maximal and verifying the node law) from which we then derive an optimal flow. This can be simply done incrementally from the feasible flow $f_h^*$ depending on the different cases illustrated in Fig.\, \ref{fig:flow path}.

\noindent
\underline{Case 1}: $n_k$ is inserted in previous graph $G$

\smallskip

All capacities on the new edges incident to of new node $n_k$ are set to $1$ while their initial flows are set to $0$. Then the flow on edge $(n_k,t)$ is obviously saturated. Two cases must be now considered:

 \begin{itemize}
    \item The node $\theta$ is on the right hand-side (see Fig.\,\ref{fig:flow path-a}). The previous capacity and flow on edge $(\theta, t)$ are decreased by one unit to match the demand value at $t$. Then, we select a node $i$
    , in the left-hand side of the bipartite graph $G$, among nodes for which the previous flow $f_{i \theta}^*$ was saturated and minimizing the cost $\gamma(i,n_k)$, thus maximizing the value of $\gamma(i,\theta)-\gamma(i,n_k)$. And by one unit we increase the flow on edge $(i,n_k)$  while decreasing that on edge $(i,\theta)$, to verify the node law at nodes $i$, $n_k$ and $\theta$. 
    
    \item The node $\theta$ is on the left hand-side (see Fig.\,\ref{fig:flow path-b}). The previous capacity and flow on edge $(s, \theta)$ are increased by one unit as well as the supply and demand values of respectively $s$ and $t$. Then we augment the flow $f_{\theta n_k}^*$ by one unit, to verify the node law at $n_k$ and $\theta$.
 \end{itemize}

\noindent
\underline{Case 2}: $n_k$ is deleted from previous graph $G$

\smallskip

We model this situation by setting the flow from $n_k$ to $t$ and the capacities on all $n_k$ incident edges to $0$. We then consider the unique node $i$ from the left-hand side that was previously connected to $n_k$ with a unit flow. We must decrement the flow on this edge to verify the node law at $n_k$. Two cases must be now considered:
\begin{itemize}
    \item The node $\theta$ is on the right hand-side (see Fig.\,\ref{fig:flow path-c}). The previous capacity and flow on edge $(\theta, t)$ are increased by one unit to match the $t$ demand value. Then, we decrease the flow on edge $(i, \theta)$ to verify the node law at $i$ and $\theta$.
    
    \item The node $\theta$ is on the left hand-side (see Fig.\,\ref{fig:flow path-d}). The previous capacity and flow on edge $(s, \theta)$ are decreased by one unit as well as the supply and demand values of respectively $s$ and $t$. If $i$ is $\theta$ then the algorithm stops. Otherwise, we can pick  a node $j$, in the right-hand side nodes, among those for which the previous flow $f_{\theta j}^*$ was saturated and minimizing the cost $\gamma(i,j)$, thus maximizing the value of $\gamma(\theta,j)-\gamma(i,j)$. By one unit we increase the flow on edge $(i,j)$ while decreasing that on edge $(\theta,j)$  to verify the node law at $i$ and $j$.
\end{itemize}

\begin{figure}[th]
\centering
\subcaptionbox{\label{fig:flow path-a}}{\includegraphics[width=5.5cm]{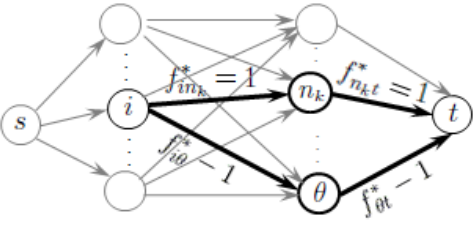}}
\hspace{1.2cm}\subcaptionbox{\label{fig:flow path-b}}{\includegraphics[width=5.5cm]{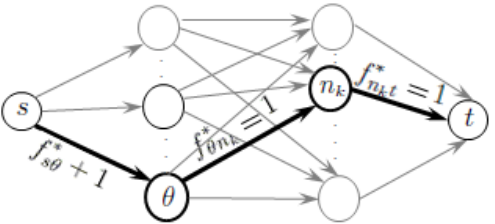}}\\
\subcaptionbox{\label{fig:flow path-c}}{\includegraphics[width=5.5cm]{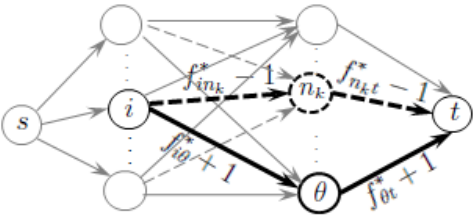}}
\hspace{1.2cm}\subcaptionbox{\label{fig:flow path-d}}{\includegraphics[width=5.5cm]{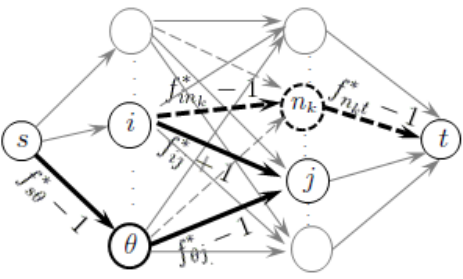}}
\caption{Examples of bipartite matching graphs $G$ (mapping forests $F_0[v]$ and $F_{h+1}[w]$) of optimal flow $f^*$ at step $h$ and which are updated at level of a node $n_k$ at step $h+1$. Bold edges highlight edges on which the flow $f^*$ is updated at step $h+1$ to get a feasible one. Elements in dashed lines represent parts of $G$ at step $h$ to be deleted at step $h+1$. The four examples differ in the $\theta$ node position (depending on the degrees of $v$ and $w$) and on the operation performed on node $n_k$. 
(a) $n_k$ is inserted in $G$ and $\Outdegree(v)>\Outdegree(w)$. (b) $n_k$ is inserted in $G$ and either $\Outdegree(v)<\Outdegree(w)$ or $\Outdegree(v)=\Outdegree(w)$ and $\theta$ has just been inserted in $G$ with incident edges of a null flow value. (c) $n_k$ is about to be deleted from $G$ and either $\Outdegree(v)>\Outdegree(w)$ or $\Outdegree(v)=\Outdegree(w)$ and $\theta$ has just been inserted in $G$ with incident edges of a null flow value. (d) $n_k$ is about to be deleted from $G$ and $\Outdegree(v)<\Outdegree(w)$.
\label{fig:flow path}}
\end{figure}

The resolutions of problems A and B complete the incremental computation of the edit-distance algorithm (a pseudo-code version of this algorithm is provided in Appendix~\ref{pseudo-alg}). We now can derive the overall computational complexity of the new incremental algorithm.

\begin{prop}
Let $T_0$ and $T_{h}$ be two trees and $T_{h+1}$ be a tree obtained from $T_{h}$ using a single edit operation. Knowing the edit-distance between $T_0$ and $T_{h}$ and the related distance information, the complexity of computing incrementally the edit distance between trees $T_0$ and $T_{h+1}$ is in $O(\Size T_0 \times \Height(T_{h+1}) \times \Outdegree(T_{h+1}) \times  \max(\Outdegree(T_{0}), \Outdegree(T_{h+1})))$.
\end{prop}

\begin{proof}
Altogether, the incremental edit distance algorithm browses the set of vertices in $T_0$ and updates the corresponding distance information for the edited node $k$ and its parents. For this, the algorithm must repeatedly solve problems A and B for all the updated pairs (see Algorithm~\ref{alg:tree-edit-distance} in Appendix~\ref{pseudo-alg}). Let us first compute the complexities of algorithms solving problems A and B as defined above.

\smallskip

The complexity of problem A depends on the worst number of loops performed by the NS algorithm starting with a given initial feasible flow solution (the flow at previous step $f^*$ that is feasible). Each NS loop transforms a feasible flow  solution into another while decreasing its total cost  until it becomes optimal.
In the graph $G$ mapping the forests rooted in $v$ and $w$, let $n$ be the $G$ node belonging to the right hand side partition and representing a subtree rooted in the ancestor vertex of the edited vertex $k$ in tree $T_{h+1}$, and let $(i,n)$ be the unique saturated edge of the flow $f_h^*$ at step $h$. According to the edit operations definition (given in Subsection~\ref{ss:zhang:operations}), the costs on the $n$ incoming edges are either incremented or decremented by one unit at step $h+1$ compared with step $h$. All other costs remain the same. Thus, at step $h+1$, either:

\begin{enumerate}
\item $f_h^\ast$ was the unique minimal solution at step $h$: \begin{enumerate}[label=\alph*.]
    \item if the cost $\gamma(i,n)$ is decremented, it certainly remains minimal;
    \item if it is incremented, $f_h^\ast$ remains minimal but it may no longer be the only one. Therefore, we can still write $f_{h+1}^\ast=f_h^\ast$.  
\end{enumerate}

\item $f_h^\ast$ was not the unique minimal solution at step $h$. Let us denote $f_{h}'^\ast$ another solution for which the flow on the unique edge $(i',n)$ is saturated, then either:  
\begin{enumerate}[label=\alph*.]
\item the cost $\gamma(i,n)$ is decremented, thus $f_h^\ast$ remains minimal and $f_{h+1}^\ast=f_h^\ast$;
\item the cost $\gamma(i,n)$ is incremented and either:
\begin{enumerate}
\item the cost $\gamma(i',n)$ is incremented, thus $f_h^\ast$ remains the minimum and $f_{h+1}^\ast=f_h^\ast$;
\item the cost $\gamma(i',n)$ is decremented, and the flow $f_{h}'^\ast$ is the new optimal solution, $f_{h+1}^\ast=f_{h}'^\ast$.
\end{enumerate}
\end{enumerate}
\end{enumerate}

As a consequence, we notice that for all these situations, the minimum flow remains unchanged between step $h$ and $h+1$, except for the unique case 2.b.ii where the NS algorithm will perform a number of loops to find the flow $f_{h}'^\ast$ from $f_h^\ast$. Each loop swaps a subset of  $f_h^\ast$ saturated flows solutions by an other subset of saturated flows in $f_{h}'^\ast$ while decreasing the total cost on $G$. In the worst case the number of loops is bounded by the supply-demand value on $G$, namely $\max(\Outdegree(v),\Outdegree(w))$. 
According to \cite{kiraly-kovacs2012, ahuja2017}, each loop takes a worst time of $O(\Outdegree(v)\times \Outdegree(w))$.  This gives an over all complexity of problem A of $O(\Outdegree(v)\times \Outdegree(w)\times \max(\Outdegree(v),\Outdegree(w)))$.

\smallskip

For problem B, recovering the flow feasibility on graph $G$ mapping the forests rooted in $v$ and $w$, as we describe below, takes a worst time of $O(\Outdegree(v)+\Outdegree(w))$.

\smallskip

Finally, the total complexity of the incremental edit distance algorithm is obtain by assembling these complexities together following Algorithm \ref{alg:tree-edit-distance}:
\begin{eqnarray*}
&& \sum_{v\in T_0}  O\left(\Outdegree(v)+\Outdegree(\Parent(k))+\sum_{w\in \Parent^+(k)} \Outdegree(v)\times \Outdegree(w)\times \max(\Outdegree(v),\Outdegree(w))\right) \\
&  =&  O\left(\Size T_0 + \Size T_0 \times \Outdegree(T_{h+1})+ \sum_{v\in T_0}\sum_{w\in \Parent^+(k)}\Outdegree(v)\times \Outdegree(w)\times \max(\Outdegree(T_{0}), \Outdegree(T_{h+1}))\right) \\
& =&  O\left(\Size T_0 + \Size T_0 \times \Outdegree(T_{h+1}) +\max(\Outdegree(T_{0}), \Outdegree(T_{h+1})) \times \sum_{v\in T_0} \Outdegree(v) \times \sum_{w\in \Parent^+(k)} \Outdegree(w)\right) \\
&=&  O\left(\Size T_0 + \Size T_0 \times \Outdegree(T_{h+1})+ \max(\Outdegree(T_{0}), \Outdegree(T_{h+1})) \times \sum_{v\in T_0}  \Outdegree(v) \times \Height(T_{h+1}) \times \Outdegree(T_{h+1}) \right)\\
&=&  O\left(\Size T_0 + \Size T_0 \times \Outdegree(T_{h+1})+ \Height(T_{h+1}) \times \Outdegree(T_{h+1}) \times \max(\Outdegree(T_{0}), \Outdegree(T_{h+1})) \times \Size T_0 \right)\\
&=& O(\Size T_0 \times \Height(T_{h+1}) \times \Outdegree(T_{h+1}) \times \max(\Outdegree(T_{0}), \Outdegree(T_{h+1}))),
\end{eqnarray*}
which states the result.
\end{proof}

\section{Simulation study}
\label{s:simu}

The goal of this simulation study is twofold: first, we aim at proving the interest of the incremental algorithm in terms of computation time to evaluate the edit distance; second, we exploit the efficiency of this incremental algorithm to investigate in depth the escape rate of the isotropic random walk on unordered trees from intensive numerical experiments.

\smallskip

We ran all the simulations presented in this section in \verb+Python3+ on a Macbook Pro laptop running OSX High Sierra, with 2.9 GHz Intel Core i7 processors and 16 GB of RAM. The classical and incremental edit distance algorithms, as well as the random walk, have been implemented as a module of the \verb+Python+ library \verb+treex+ \cite{Azais2019}.

\subsection{Computation time improvement}
\label{ss:comptimes}

Starting from two random trees $T$ and $T'$ for which the edit distance has been pre-evaluated, we have computed $\Dist(T,o(T'))$ using both CUTED algorithm and the incremental algorithm developed in this paper, for $o$ picked at random in the set of all possible operations applicable to $T'$, distinguishing the type of operation: add leaf, del leaf, add internal node, or del internal node. The boxplots of computations times required to evaluate $\Dist(T,o(T'))$ are presented in Fig.\,\ref{fig:comptimes:sizes} for different trees sizes. They have been estimated from $1\,000$ independent experiments for each size. The average computation times (over all the possible operations) as well as the average rates of improvement are given in Tab.\,\ref{tab:comptime:t}.

\begin{figure}[th]
\centering\includegraphics[width=14cm]{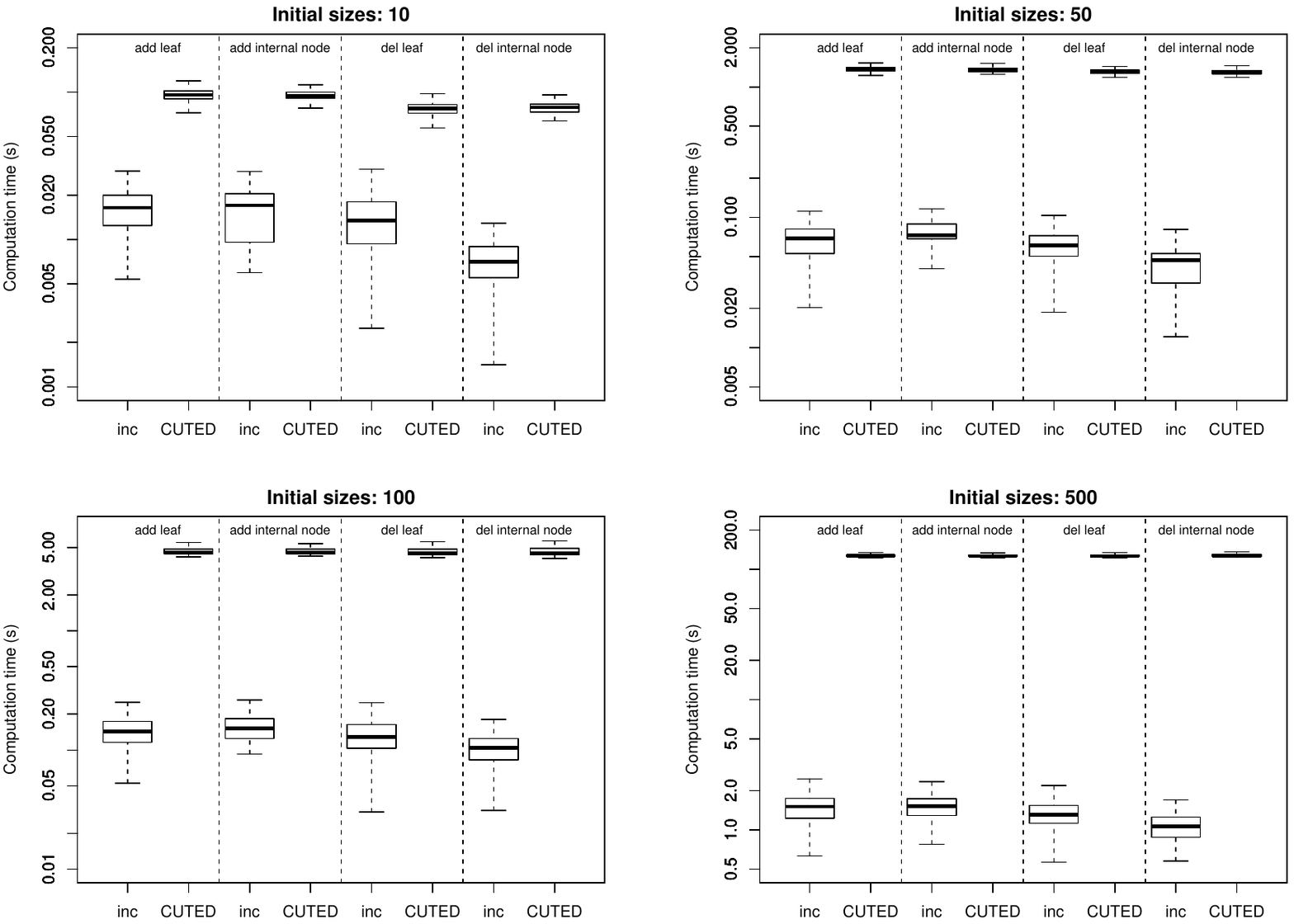}
\caption{Computation times (logarithmic scale) required for computing the edit distance between two trees of different sizes (top left: $10$, top right: $50$, bottom left: $100$, and bottom right: $500$) after editing the second one with one of the 4 possible operations, with the algorithm from the literature (CUTED) and the incremental one developed in this paper. Because of the logarithmic scale, the boxplots from CUTED algorithm seem to be narrower than from the incremental algorithm, while they are actually more scattered.}
\label{fig:comptimes:sizes}
\end{figure}

\begin{table}[ht]
\centering
\begin{tabular}{c|cccc}
Trees size & 10 & 50 & 100 & 500 \\ \hline
\begin{tabular}{@{}c@{}}Computation time (s) \\
CUTED algorithm
\end{tabular} & 0.08378765 & 1.337326 & 4.90168  & 127.9237 \\ \hline
\begin{tabular}{@{}c@{}}Computation time (s) \\
Incremental algorithm
\end{tabular} & 0.01455773 & 0.06493733 &  0.1449104 & 1.425516 \\ \hline
\begin{tabular}{@{}c@{}}Ratio of computation times \\ ``CUTED/incremental''\end{tabular} & 6.99916 & 23.1193 & 38.00966 & 96.71445 \\ \hline
\end{tabular}
\vspace{0.5cm}
\caption{Average computation time (s) and rate of improvement for the incremental algorithm compared to CUTED algorithm for trees of different sizes. Estimates from $1\,000$ independent replicates. This table compiles the data of Fig.\,\ref{fig:comptimes:sizes}, which present them for the different edit operations.}
\label{tab:comptime:t}
\end{table}

From Fig.\,\ref{fig:comptimes:sizes}, one can observe that the incremental algorithm is faster than CUTED algorithm, and the larger the trees, the (relatively) faster it is. One also remarks that the computation times are more predictable (in the sense that the boxplots are narrow) for the incremental algorithm. The gain in terms of computation time is already significant for trees of size $10$, since the rate of improvement is of order $7$ (from $15$ to $84$\,ms on average). For trees of size $500$, one sees a gain of almost two orders of magnitude (see Tab.\,\ref{tab:comptime:t}). These numerical results show the great efficiency of our approach. In particular, simulating the trajectory of a random walk while evaluating the end-to-end distance is much faster using the incremental algorithm: for example, $14.9$ (CUTED algorithm) vs. $1.8$\,s (incremental algorithm) for $200$ steps starting from a tree of size $10$. See Tab.\,\ref{tab:comptime:rw} for average computation times of random walk simulations.

\begin{table}[ht]
\centering
\begin{tabular}{c|cccc}
 & 25 steps & 50 steps & 100 steps & 200 steps \\ \hline
\begin{tabular}{@{}c@{}}Computation time (s) \\
CUTED algorithm
\end{tabular} & 1.591511  &3.199535 &  5.190579 & 14.883755 \\ \hline
\begin{tabular}{@{}c@{}}Computation time (s) \\
Incremental algorithm
\end{tabular} & 0.2946835 & 0.4631565 & 0.7088750 & 1.7541514 \\ \hline
\end{tabular}
\vspace{0.5cm}
\caption{Average computation time (s) to simulate and evaluate the end-to-end distance of a random walk starting from a tree of size $10$. Estimates from $100$ independent replicates.}
\label{tab:comptime:rw}
\end{table}

\subsection{Sharp estimation of average escape rate}
\label{ss:sharpesc}

Fig.\,\ref{fig:random:traj:av} presents a sample trajectory starting from an initial tree of size $10$ within a time window of $10\,000$ steps. The 5 characteristics introduced in Subsection~\ref{ss:def:trees}, namely the end-to-end distance, the size, the height, the outdegree, and the Strahler number, as well as their average behavior estimated from $1\,000$ replicates, are given. For all of them, one observes a slowing escape towards infinity, which results in concave curves. The aim of this subsection is to estimate, as precisely as possible, the average escape rate of the end-to-end distance thanks to intensive simulations.

\begin{figure}[h]
\centering\includegraphics[width=14cm]{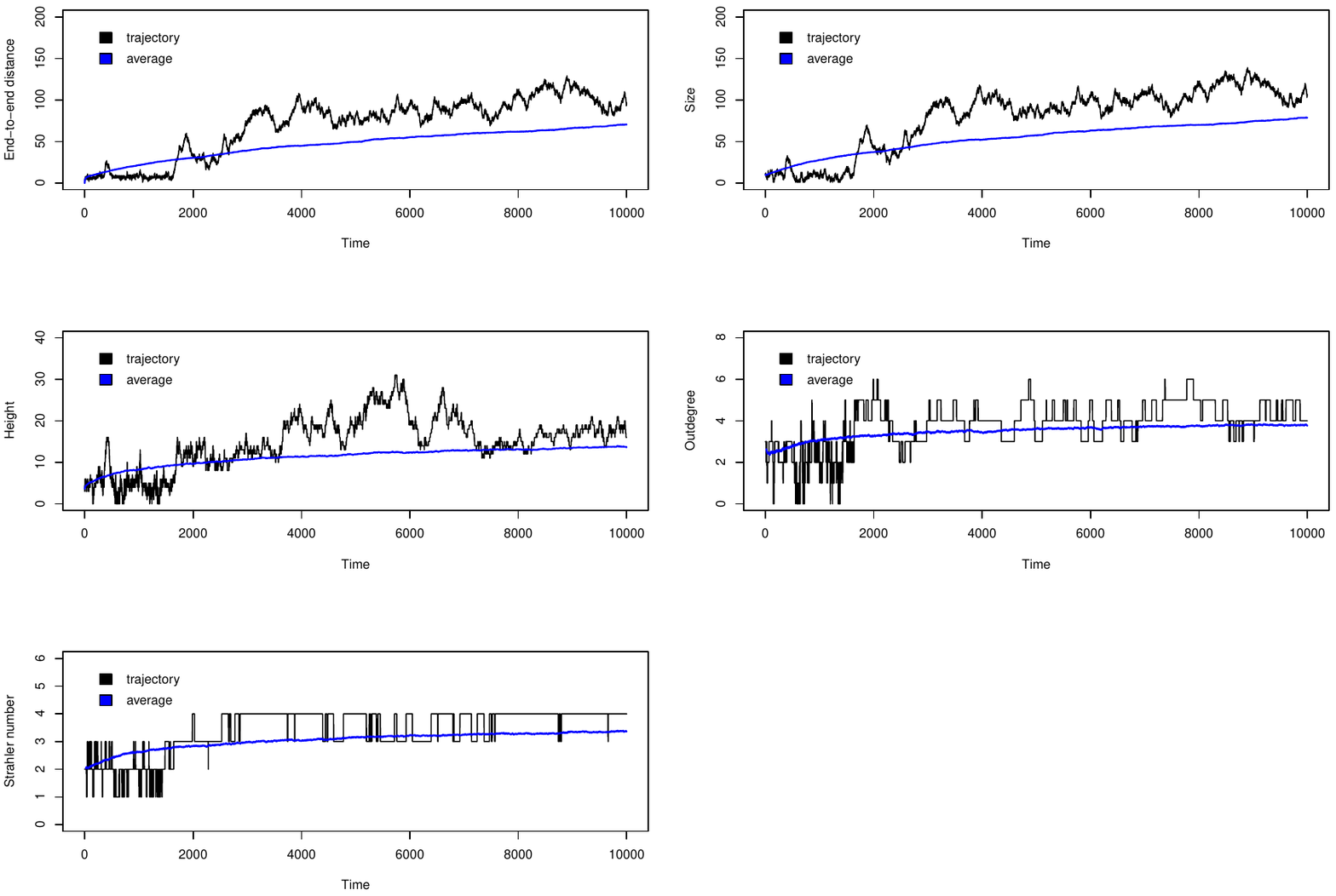}
\caption{5 characteristics of a sample trajectory of the random walk starting from an initial tree of size $10$ up to time $10\,000$ (in black) and average curves (in blue, estimated from $1\,000$ independent replicates): end-to-end distance (top left), size (top right), height (middle left), outdegree (middle right), and Strahler number (bottom left). See Subsection~\ref{ss:def:trees} for definitions.}
\label{fig:random:traj:av}
\end{figure}

\smallskip

The end-to-end distance can be studied through the evolution of the size. Given any trees $T$ and $T'$, the distance between them can be bounded from their sizes as
\begin{equation}\label{eq:size:dist}
    |\Size T'-\Size T| \leq \Dist(T,T') \leq \Size T'+\Size T .
\end{equation}
The upper-bound corresponds to the worst case where all the vertices of $T$ have been deleted before adding all the vertices of $T'$. The lower bound is obtained in the very favorable case where $T'$ can be obtained from $T$ by only adding (or deleting) vertices. From \eqref{eq:size:dist} and given a random walk on trees $(T_h)_{h\geq0}$, we obtain the following relation between the size and end-to-end distance processes,
\begin{equation}\label{eq:rel:size:dist}
    1-\frac{\Size T_0}{\Size T_t} \leq \frac{\Dist(T_0,T_t)}{\Size T_t} \leq 1+\frac{\Size T_0}{\Size T_t}.
\end{equation}
However, it should be noted that the size process $(\Size T_t)_{t\geq0}$ is a simple random walk on integers reflected at the barrier $1$, i.e. a $\mathbb{N}^\ast$-valued process such that,
$$
    \forall\,n\geq2,~\mathbb{P}(\Size T_{h+1} = n+1|\Size T_h = n) = \mathbb{P}(\Size T_{h+1} = n-1|\Size T_h = n) = \frac{1}{2},
$$
and $\mathbb{P}(\Size T_{h+1}=2|\Size T_h = 1) = 1$. The average escape rate of a (reflected in $0$) random walk $(W_h)_{h\geq0}$ has been studied in the literature. Taking $p_1 = 1$ in \cite[eq.\,(1.2)]{Katriel2011AsymptoticBO}, we obtain that, when $h$ goes to infinity, $\mathbb{E}[W_h|W_0 = x] \sim \rho(x,h)$, where
\begin{equation}\label{eq:esc:srw}
    \rho(x,h) = \sqrt{\frac{2h}{\pi}} + \frac{x^2}{\sqrt{2\pi h}}.
\end{equation}
Together with \eqref{eq:rel:size:dist}, and using the fact that $\Size T_h$ can be expressed as $W_h+1$, this proves that, when $h$ goes to infinity,
\begin{equation}\label{eq:escrate}
    \mathbb{E}[\Dist(T_0,T_h) | \Size T_0 = x] \sim \rho(x,h).
\end{equation}
We aim to improve our empirical knowledge of this rate thanks to intensive simulations of random walks, in particular through the study of the dependency on the initial condition $x$.

\smallskip

In order to refine the escape rate \eqref{eq:escrate}, we have performed intensive simulations of random walks, within a time window of $10\,000$ steps repeated $1\,000$ times to achieve accurate average behaviors, starting from $9$ possible initial sizes (from $2$ to $10$ nodes). Consequently, $90\text{M}$ evaluations of distances have been completed to conduct this study.

\smallskip

Fig.\,\ref{fig:av:dist:size} presents the average size and end-to-end distance processes starting from initial trees of different sizes. We remark that the theoretical rate \eqref{eq:esc:srw} is accurate for the size process as expected, but not for the end-to-end distance, in particular when the size of the initial condition $\Size T_0$ increases. It is not surprising in regards to \eqref{eq:rel:size:dist} that introduces an error of order $O(\Size T_0/\Size T_h)$.

\begin{figure}[h]
\centering\includegraphics[width=14cm]{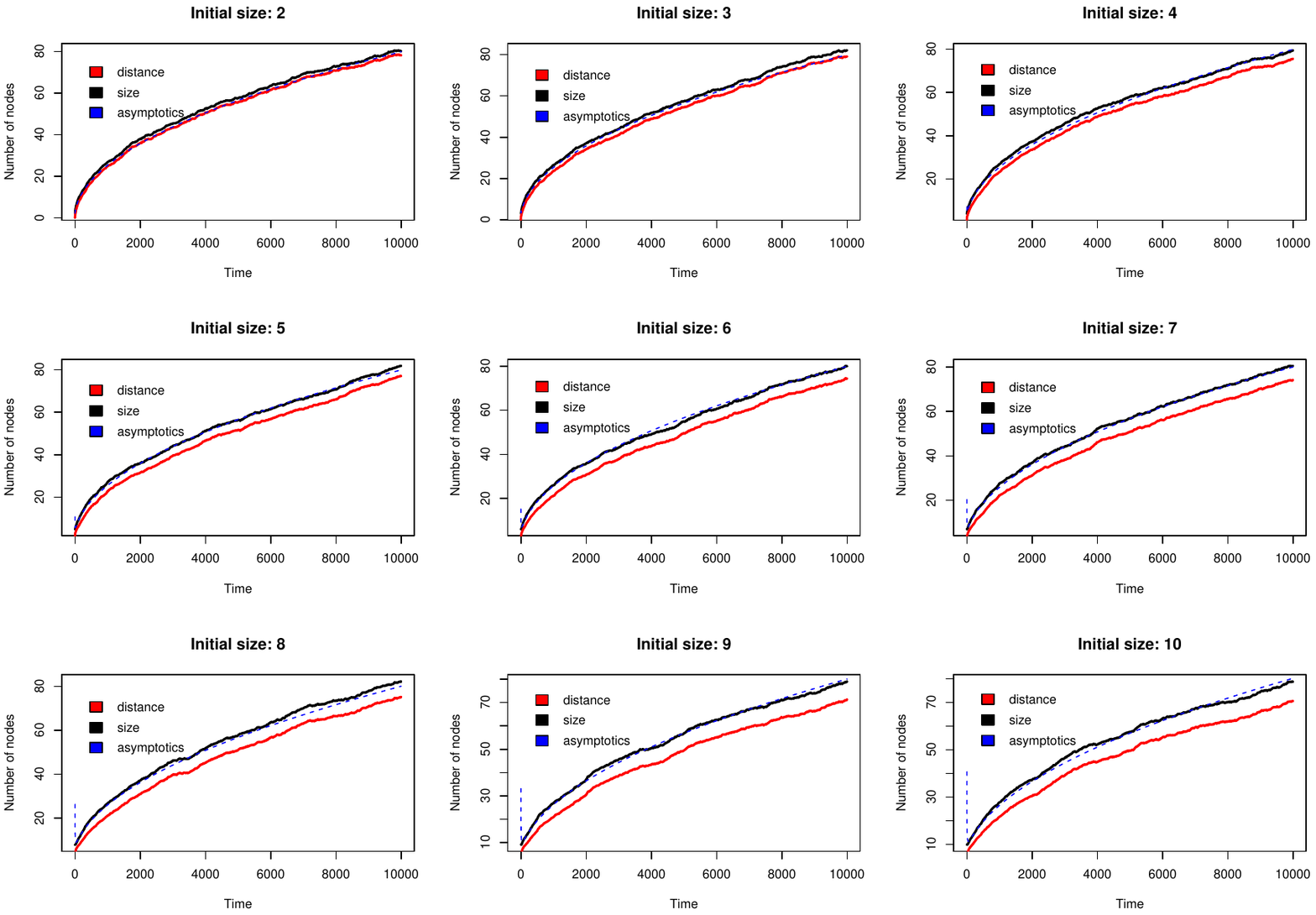}
\caption{Average size process $(\Size T_h)_{h\geq0}$ (in black) and average end-to-end distance process $(\Dist(T_0,T_h))_{h\geq0}$ (in red) starting from trees of size between $2$ and $10$. Each curve has been estimated from $1\,000$ independent replicates of the random walk. The theoretical asymptotics $\rho(\Size T_0,h)$, see \eqref{eq:esc:srw}, of the size process has been added in blue dashed lines.}
\label{fig:av:dist:size}
\end{figure}

Instead of comparing $\Dist(T_0,T_h)$ with $\Size T_h$, we propose in Fig.\,\ref{fig:ratio:dist:size} to investigate the approximation accuracy of the two bounds given in \eqref{eq:size:dist}, i.e. compare $\Dist(T_0,T_h)$ with the lower-bound $\Size T_h-\Size T_0$ and with the upper-bound $\Size T_h+\Size T_0$. The numerical simulations empirically show that $\Size T_h-\Size T_0$ is the best approximation of the the distance process, which tends to prove that the random walk keeps a long-term memory of the starting tree. Indeed and as aforementioned, $\Size T_h-\Size T_0$ is equal to the distance when nodes have been only added to the initial structure, without removing any part of it.

\begin{figure}[ht]
\centering\includegraphics[width=14cm]{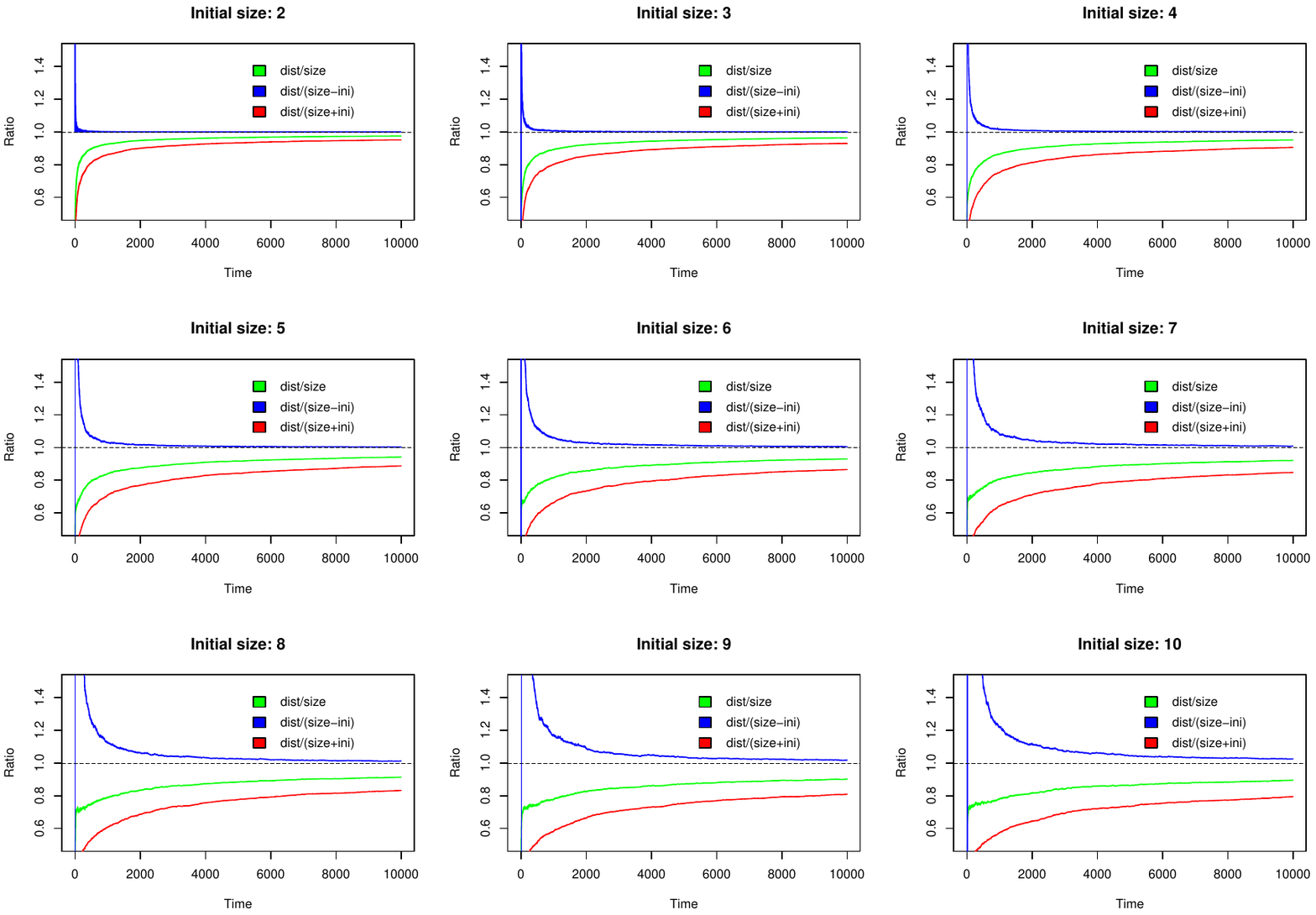}
\caption{Ratio of average end-to-end distance process $(\Dist(T_0,T_h))_{h\geq0}$ over (i) average size process $(\Size T_h)_{h\geq0}$ (in green), (ii) average size process minus initial size $(\Size T_h-\Size T_0)_{h\geq0}$ (in blue), and (iii) average size process plus initial size $(\Size T_h+\Size T_0)_{h\geq0}$ (in red).}
\label{fig:ratio:dist:size}
\end{figure}

\smallskip

We also remark from Fig.\,\ref{fig:ratio:dist:size} that $\Dist(T_0,T_h)$ is less accurately estimated by $\Size T_h-\Size T_0$ when $\Size T_0$ increases, which says that the memory of the starting tree is less preserved when it is larger. We propose to estimate this approximation error through a polynomial fitting presented in Fig.\,\ref{fig:adjust:rate}: (i) at each time step $h$, we observe that the ratio $\mathbb{E}[\Dist(T_0,T_h) | \Size T_0 = x] / (\mathbb{E}[\Size T_h | \Size T_0 = x]-\Size T_0)$ is surprisingly well-captured by a function of the form $1+\alpha_t\Size T_0^2$ (3 examples are given in Fig.\,\ref{fig:adjust:rate} top left, top right, and bottom left); (ii) we show that the regression coefficients $\alpha_h$ behave as $\beta/h^{5/4}$, with $\beta=17.83953$ (see Fig.\,\ref{fig:adjust:rate} bottom right). As a consequence, we obtain this estimated but sharp rate,
$$
    \frac{\mathbb{E}[\Dist(T_0,T_h) | \Size T_0 = x]}{\mathbb{E}[\Size T_h | \Size T_0 = x]-\Size T_0} \sim 1 + \frac{\beta\,\Size T_0^2}{h^{5/4}}.
$$
Finally, together with all of the above, our estimated escape rate of the end-to-end distance process is given by
\begin{equation}\label{eq:est:sharp:rate}
    \mathbb{E}[\Dist(T_0,T_h) | \Size T_0 = x] \sim \sqrt{\frac{2h}{\pi}}-\Size T_0 + \frac{\Size T_0^2}{\sqrt{2\pi h}} + \frac{\beta\, \Size T_0^2}{h^{3/4}}\sqrt{\frac{2}{\pi}} -  \frac{\beta\,\Size T_0^3}{h^{5/4}}
    + \frac{\beta\,\Size T_0^4}{\sqrt{2\pi} \, h^{7/4}} ,
\end{equation}
which accuracy is illustrated in Fig.\,\ref{fig:sharp:escape:rate}.

\begin{figure}[ht]
\centering\includegraphics[width=14cm]{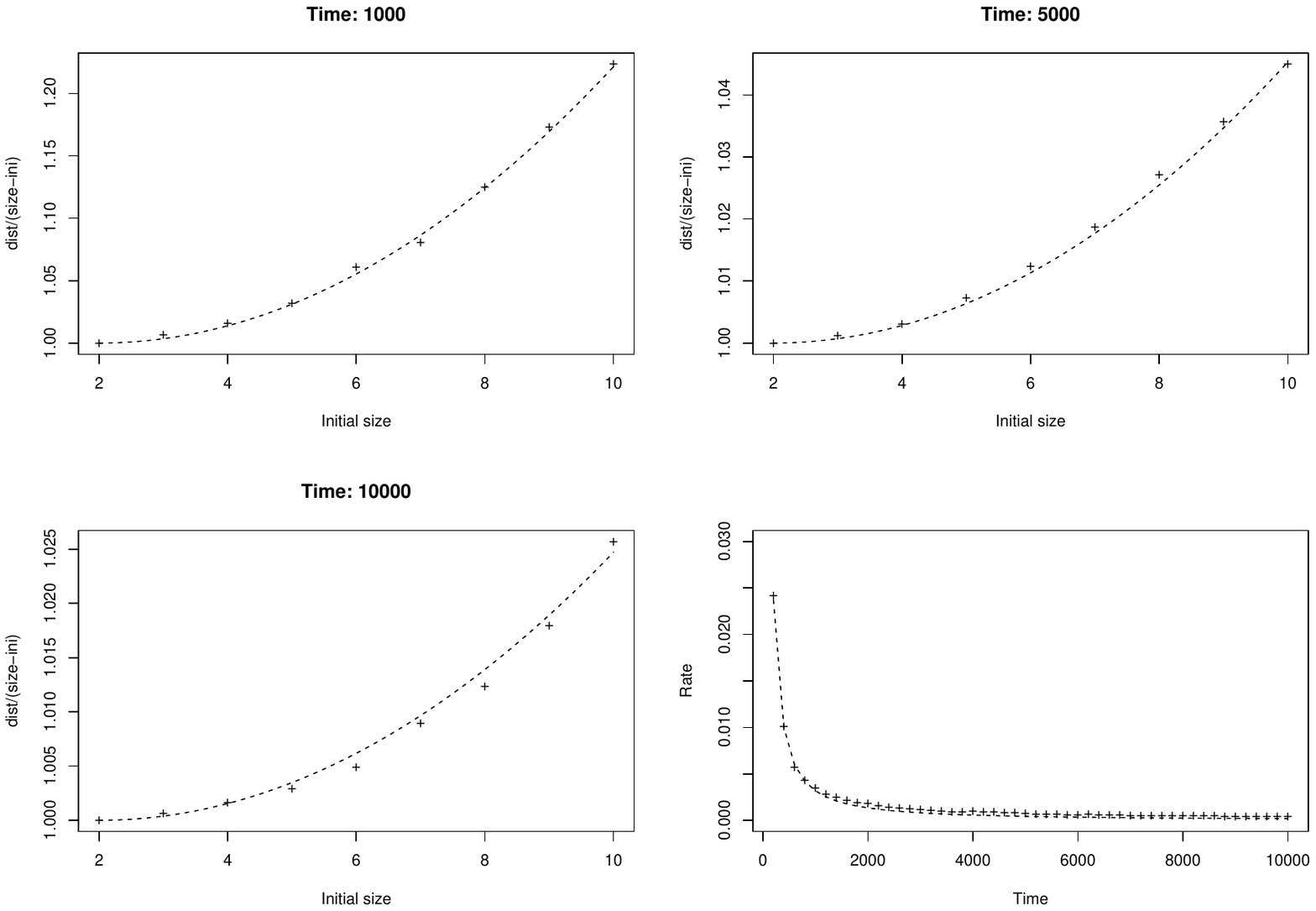}
\caption{Ratio of average end-to-end distance process $(\Dist(T_0,T_h))_{h\geq0}$ over average size process minus initial size $(\Size T_h-\Size T_0)_{h\geq0}$ as a function of the initial tree size (crosses) and their quadratic fitting of the form $1+\alpha_h\Size T_0^2$ (dashed lines), at 3 time steps (top left: $h=1\,000$, top right: $h=5\,000$, and bottom left: $h=10\,000$). Fitting coefficients $\alpha_h$ as a function of time $h$ (crosses) and fitting of the form $\beta/h^{5/4}$ (dashed line) with $\beta = 17.83953$ (bottom right).}
\label{fig:adjust:rate}
\end{figure}

\begin{figure}[ht]
\centering\includegraphics[width=14cm]{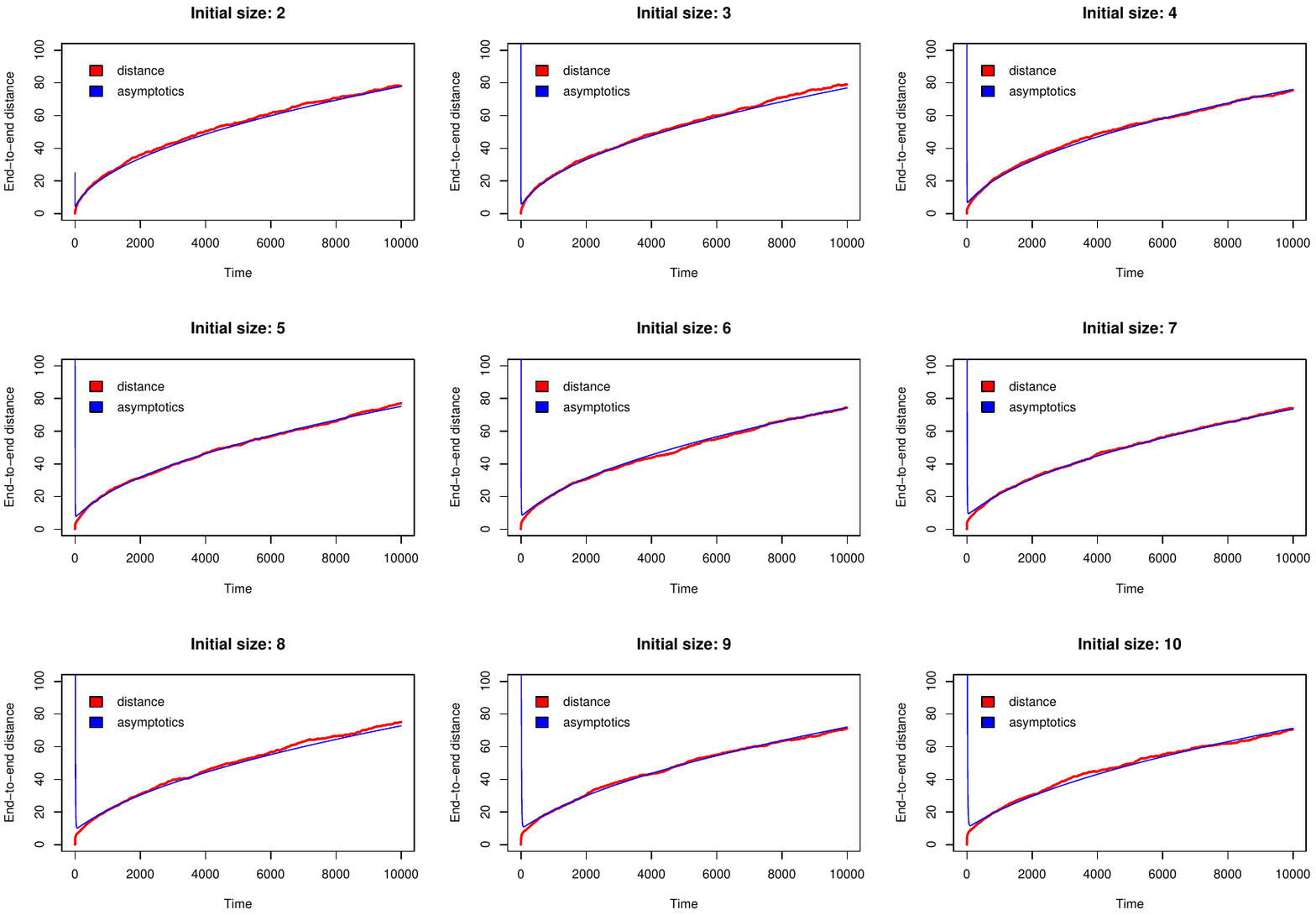}
\caption{Average end-to-end distance process $(\Dist(T_0,T_h))_{h\geq0}$ (in red) starting from trees of size between $2$ and $10$, and estimated asymptotics (in blue) given in \eqref{eq:est:sharp:rate}.}
\label{fig:sharp:escape:rate}
\end{figure}

\section{Discussion}
\label{s:discussion}

In the present paper, we have developed an incremental algorithm to quickly evaluate the edit distance \cite{zhang96} to the origin along a random walk on the space of unordered trees. As an application, a method for accurately estimating the escape rate of the balanced random walk has been derived, taking into account the theoretical asymptotic behaviour of the size process. Other aspects of this work that may be of interest are discussed below.

\subsection{Isotropic random walk}
\label{ss:isotropic}

The random walk considered in this article is balanced in the following sense: from any tree, the probability to reach a bigger neighbour is equal to the probability to reach a smaller neighbour. Another option would be to consider the isotropic random walk, which assigns the same probability to all the feasible edit operations. Incremental CUTED algorithm and the escape rate estimation method could be applied to assess its asymptotic behaviour. As remarked in Subsection~\ref{ss:isorw}, the number of neighbours of size $n+1$ of a tree of size $n$ is in general larger than the number of neighbours of size $n-1$, and the gap tends to increase with $n$. The isotropic random walk is thus expected to escape very fastly. The reflection on the single-node tree is therefore very rare, whereas it is an intriguing property of the balanced random walk, which in particular links it to the random walk on $\mathbb{N}^\ast$ as explained in Subsection~\ref{ss:sharpesc}.

\subsection{Applicability of incremental CUTED}
\label{ss:applicability}

The incremental CUTED algorithm can be applied in all problems where one seeks to evaluate the edit distance to a reference along a trajectory on the space of trees made of elementary moves. This is the case, for instance, when one aims to solve optimization problems by heuristic algorithms such as simulated annealing \cite{sa}. Two examples are provided below.

\smallskip

One typical problem in statistics of tree data is to estimate the barycenter of a set of $n$ trees $\{T_i\}_{1\leq i\leq n}$, i.e. the minimum of the function $\Delta:\theta\mapsto \sum_{i=1}^n \Dist(\theta,T_i)$. Simulated annealing consists in building a Markov chain $(\theta_k)_{k\geq0}$ that converges towards the targeted optimum in the following way:
\begin{itemize}
\item From a tree $\theta_k$, select a random feasible edit operation: it defines a new tree $\widehat{\theta}_{k+1}$ ;
\item Evaluate the difference $d = \Delta(\widehat{\theta}_{k+1})-\Delta(\theta_k)$ ;
\item Accept $\widehat{\theta}_{k+1}$ with probability $p = \min(1,\exp(-\beta_k d))$), i.e.
$$\theta_{k+1} = \left\{\begin{array}{cl}
\widehat{\theta}_{k+1} & \text{with probability $p$,}\\
\theta_k & \text{with probability $1-p$,}
\end{array}\right.$$
where the sequence $(\beta_k)_{k\geq0}$ increases slowly.
\end{itemize}
The edit distance to the $n$ reference trees must be evaluated at each step of the trajectory, which can be efficiently performed by the incremental CUTED algorithm.

\smallskip

A similar problem occurs in lossy compression of trees. Self-nested trees are unordered trees which all subtrees of the same height are isomorphic. Because of the high redundancy of their subtrees, they achieve remarkable compression properties \cite{selfnested}. Lossy compression of an unordered tree $T$ can be performed through its approximation by a self-nested tree. In other words, one seeks the minimum of the function $\Delta:\theta\mapsto\Dist(\theta,T)$ on the subspace of self-nested trees. Different strategies have been developed in recent years \cite{selfnested,azais:hal-01584078,azais:hal-01294013} but the problem in its generality remains open. The applicability of heuristics like simulated annealing highly depends on the computation time required to evaluate the edit distance along the trajectory, which shows again the interest of the incremental CUTED algorithm.

\subsection{Labeled trees}

The method and algorithms developed in this paper only deal with topological trees, without labels on their nodes. They are of great interest both for fundamental questions in mathematics and computer science they raise and in the applications (see Subsection~\ref{ss:applicability}). Furthermore, in applications of labeled trees, topology and labels can be handled simultaneously (for instance through an edit distance on labeled trees) or separately, by considering first the topology then the labels. In the second case, topology can be taken into account through an edit distance on unlabeled trees.

\smallskip

While the space of topological trees inherits a canonical graph structure (and thus a canonical random walk), this is not the case for labeled trees, where both size and connectivity depend strongly on the label space. Choosing a finite set, $\mathbb{N}$, $\mathbb{Z}$ or $\mathbb{Z}^p$ for the set of labels induces therefore very different random walks. In addition, the elementary operation that adds a node to a tree is no longer canonical because its label has to be selected from a specified distribution that will have an effect on the long-term behaviour. For instance, if the set of labels is $\mathbb{N}$, the reflection on the barrier $0$ will occur more often if the initial label is small. Finally, the escape rate can be expected to grow rapidly with the size of the label space. The escape rate of the random walk on topological trees, estimated in this paper, forms the baseline of the behaviour of random walks on labeled trees.

\smallskip

This being stated, the incremental CUTED algorithm derived in the paper can be adapted to take labels into account. For labeled trees, one can consider the following 5 elementary editing operations:
\begin{itemize}
\item Add a leaf with a given label;
\item Remove a leaf;
\item Add an internal node with a given label;
\item Remove an internal node;
\item Substitute the label of a node.
\end{itemize}
For the first 4 operations, the only difference with unlabeled trees (see Subsection~\ref{ss:zhang:operations}) is that the label of the new node has to be specified. The incremental computation of the distance proceeds in the same way as for unlabeled trees, where the computational process requires the constrained edit distance between labeled trees \cite{zhang96}. For the last editing operation, the architecture of the new tree $T_{h+1}$ remains unchanged compared with $T_h$: only the label of the impacted node $k$ is modified. This implies changes in the mapping cost of $M^\ast_h(v,w)$ for all nodes $w$ of $T_{h+1}$ along the ancestors of $k$ (as in Fig.\,\ref{fig:changes-all-cases-b}). For this case, according to the new costs, we need to solve Problem A to compute the new forest mapping $M^\ast_{h+1}(v,w)$ starting from the mapping $M^\ast_h(v,w)$, for $w$ in $\{k\}\cup\Parent^+(k)$ and for all $v$ in the reference tree $T_0$.

\section*{Acknowledgments}
The authors would like to thank an anonymous reviewer for his/her valuable comments that helped them to significantly improve the first version of the paper.

{\small
\bibliographystyle{apalike}  
\bibliography{references}
}

\bigskip

\appendix

\section{Pseudo-code of the incremental edit distance algorithm} \label{pseudo-alg}

\begin{algorithm}
\SetKwBlock{Function}{Function incremental edit distance}{end}
\Function{
\SetKwInOut{Input}{Input}
\SetKwInOut{Output}{Output}
\DontPrintSemicolon
\Input{\;
Two unordered rooted trees $T_{0}$ and $T_{h+1}$\; 
The distance information computed for $\delta(T_0,T_h)$\; 
The edited vertex $k$ in $T_{h+1}$}
\Output{\;
The new distance information computed for $\delta(T_0,T_{h+1})$.}
\BlankLine

\For {$v \in T_0$ }{
 Compute the distance information at $(v,k)$ according to the type of edit operation on $k$ (see Subsection \ref{ss:inczhang})\;
 Update the distance information at $(v,\Parent(k))$:\;
 $\quad$ solve problem B (to restore a valid forest mapping between $F_{h+1}[v]$ and $F_{h+1}[\Parent(k)]$)\;
 $\quad$ solve problem A (to restore an optimal forest mapping between $F_{h+1}[v]$ and $F_{h+1}[\Parent(k)]$)\;
\For {$w \in \Parent^+(\Parent(k))$}{
 Update the distance information at $(v,w)$: \;
 $\quad$ solve problem A to restore an optimal forest mapping between $F_{h+1}[v]$ and $F_{h+1}[w]$)
}
}
\Return The distance information of $\delta(T_0,T_{h+1})$\;
}
\caption{The incremental computation of edit distance between two trees $T_{0}$ and $T_{h+1}$. \label{alg:tree-edit-distance}}
\end{algorithm}

\end{document}